\documentclass{elsarticle}

\usepackage{lineno,hyperref}
\modulolinenumbers[5]

\usepackage[american]{babel}	%English hyphenation etc.
\usepackage[utf8]{inputenc}		%Input encoding
\usepackage[T1]{fontenc}		%Font encoding
\usepackage[english=american]{csquotes}	%Quotation marks ' '
\usepackage{graphicx}			%Include graphics
\usepackage{epstopdf}
\usepackage{subfig}				%Subfigures
\usepackage{marvosym}			%Special symbols
\usepackage{amsmath}			%Mathmode
\usepackage{amssymb}
\usepackage{amsthm}
\usepackage{mathtools}
\usepackage{color}
\usepackage[linesnumbered,ruled]{algorithm2e}

\usepackage{changes}
\makeatletter
\@namedef{Changes@AuthorColor}{red}
\colorlet{Changes@Color}{red}
\makeatother

\usepackage{calligra}
\DeclareMathAlphabet{\mathcalligra}{T1}{calligra}{m}{n} % like mathcal, but for lowercase letters

\DeclareMathOperator*{\argmin}{argmin}

\theoremstyle{definition}
\newtheorem{example}{Example}%[subsection]

\newtheorem{remark}{Remark}

\theoremstyle{plain}
\newtheorem{lemma}{Lemma}
\newtheorem{theorem}{Theorem}%[section]
\newtheorem{proposition}{Proposition}
\newtheorem{corollary}{Corollary}

\newcommand{\red}{\textcolor{black}}
\newcommand{\blue}{\textcolor{black}}

\begin{document}

\begin{frontmatter}
\title{Combinatorial views on persistent characters in phylogenetics}

\author{Kristina Wicke}

\author{Mareike Fischer\corref{cor1}}
\ead{email@mareikefischer.de}
\cortext[cor1]{Corresponding author}

\address{Institute of Mathematics and Computer Science, University of Greifswald, Greifswald, Germany}

\begin{abstract} 
The so-called binary perfect phylogeny with persistent characters has recently been thoroughly studied in computational biology as it is less restrictive than the well-known binary perfect phylogeny. Here, we focus on the notion of (binary) persistent characters, i.e. characters that can be realized on a phylogenetic tree by at most one $0 \rightarrow 1$ transition followed by at most one $1 \rightarrow 0$ transition in the tree, and analyze these characters under different aspects. 
First, we illustrate the connection between persistent characters and Maximum Parsimony, where we characterize persistent characters in terms of the first phase of the famous Fitch algorithm.
Afterwards we focus on the number of persistent characters for a given phylogenetic tree. We show that this number solely depends on the balance of the tree. To be precise, we develop a formula for counting the number of persistent characters for a given phylogenetic tree based on an index of tree balance, namely the Sackin index. 
Lastly, we consider the question of how many (carefully chosen) binary characters together with their persistence status are needed to uniquely determine a phylogenetic tree and provide an upper bound for the number of characters needed. 
\end{abstract}

\end{frontmatter}

%\linenumbers

\section{Introduction}
Reconstructing the evolutionary history of a set of species based on so-called characters is a central goal in evolutionary biology. A character is usually seen as some \enquote{characteristic} of a species and might for example be of morphological nature (e.g. a character could describe the number of teeth in different species) or of genetic nature (i.e. a genetic character could describe the nucleotide at a particular position in the DNA). Given a set $X$ of species together with a set of characters, the overall aim is now to find a phylogenetic tree with leaf set $X$ that explains the evolution of the characters associated with the species in $X$. However, there are several evolutionary models that explain how exactly a character could have evolved on a tree from some early ancestor (the root of the tree) to the present day species (the leaves of the tree), ranging from very restrictive ones (e.g. the perfect phylogeny model (\citet{Baca2001})) to more general models.
Here, we focus on binary characters and consider a model that has recently caught attention in the literature, namely the so-called \emph{binary perfect phylogeny with persistent characters}. This model is less restrictive than the perfect phylogeny model, but more restrictive than for example the Dollo Parsimony (\citet{Dollo}). 
\red{While the perfect phylogeny model is based on the assumption that a character can only be gained once throughout the course of evolution (i.e. there is no parallel evolution) and cannot be lost once it is gained (i.e. evolution cannot be reversed and there are no back mutations), the binary perfect phylogeny model with persistent characters assumes that a character can only be gained once but may also be lost once. This scenario is often more realistic than the perfect phylogeny model, for example, in the context of explaining the evolution of tumor phylogenies, where there is strong evidence for back mutations as well as parallel mutations (\citet{Kuipers2017}). The persistent phylogeny model and generalizations of it are thus an important tool in the context of reconstructing tumor phylogenies and are a topical and active area of research (\citet{Bonizzoni2017, Bonizzoni2019,El-Kebir2018}). 
In this note, we focus less on applications of the persistent phylogeny model, but more on the combinatorial structure underlying  it. More precisely, we study the combinatorial properties of characters consistent with the persistent phylogeny model, namely  so-called \emph{persistent} characters.}
A binary character is called persistent, if it can be realized on a phylogenetic tree by at most one $0 \rightarrow 1$ transition followed by at most one $1 \rightarrow 0$ transition. 
The problem of reconstructing a \emph{persistent phylogeny} for a set of species together with a set of binary characters if it exists, i.e. reconstructing an evolutionary tree on which all the given characters are persistent, is called the \emph{Persistent Phylogeny Problem} (PPP) and has been thoroughly studied (\citet{Bonizzoni2012, Bonizzoni2014, Bonizzoni2017b, Bonizzoni2017}). In particular it was shown in \citet{Bonizzoni2017b} that the Persistent Phylogeny Problem can be solved in polynomial time. 
Here, we will not be concerned with algorithmic questions regarding the Persistent Phylogeny Problem, but rather consider persistent characters from a combinatorial perspective. We start by illustrating a connection between persistent characters and  Maximum Parsimony, in particular concerning the first phase of the so-called Fitch algorithm. We then analyze the number of binary characters that are persistent on a given phylogenetic tree and relate this number to an index of tree balance, namely the Sackin index (\citet{Sackin}). 
We will see that the more balanced a tree is, the fewer binary characters are persistent on it.
Lastly, we consider the question of how many (carefully chosen) characters together with their persistence status (i.e. information about whether the characters are persistent or not) are needed to uniquely determine a phylogenetic tree. 
This question was posed by Prof. Mike Steel as part of the \enquote{Kaikoura 2014 Challenges} at the Kaikoura 2014 Workshop, a satellite meeting of the Annual New Zealand Phylogenomics Meeting (\citet{Kaikoura2014} and \url{http://www.math.canterbury.ac.nz/bio/events/kaikoura2014/files/kaikoura-problems.pdf}). We partially answer this question by providing an upper bound for this number.

\section{Preliminaries}
Before we can start to discuss the notion of persistent characters in detail, we first need to introduce all concepts used in this manuscript.

\subsection*{Trees and characters}
A \emph{tree} $T$ is a connected, acyclic graph with node set $V$ and edge set $E$. We use $V_L$ in order to denote the set of \emph{leaves} of a tree and $\mathring{V}$ to denote the set of \emph{inner nodes}. A tree is called \emph{rooted} if there is a designated root node $\rho$. Otherwise it is called \emph{unrooted}. Moreover, a tree is called \emph{rooted binary} if the root has degree 2 and all other non-leaf nodes have degree 3. For persistence, we have to add an extra edge to the root $\rho$, which we call {\em root edge}, and whose second endnode we call $\rho'$. 
Moreover, a \emph{phylogenetic $X$-tree} is a tuple $\mathcal{T} = (T, \phi)$, where $\phi$ is a bijection from $V_L$ to $X$. $T$ is often referred to as the \textit{topology} or \textit{tree shape} of $\mathcal{T} = (T, \phi)$ and $X$ is called the \emph{taxon set} of $\mathcal{T}$. Note that we refer to $T$ instead of $\mathcal{T}$ whenever the specific leaf labeling $\phi$ is irrelevant for our considerations. Throughout this manuscript, when we refer to trees, we always mean rooted binary phylogenetic trees (possibly with an additional root edge) on a taxon set $X$, and without loss of generality assume that $X=\{1, \ldots, n\}$. Furthermore, we implicitly assume that all trees are \emph{directed} from the root to the leaves and for an edge $e=(u,v)$ we call $u$ the \emph{direct ancestor} or \emph{parent} of $v$ (and $v$ the \emph{direct descendant} or \emph{child} of $u$). Alternatively, we sometimes call $u$ the \emph{source node} of the edge $(u,v)$. 
Furthermore, we call the two subtrees rooted at an internal node $u$ the two \emph{maximal pending subtrees} of $u$.

Moreover, recall that two leaves $v$ and $w$ are said to form a \emph{cherry} $[v,w]$, if $v$ and $w$ have the same parent. 

Now, let $T$ be a rooted binary tree with root $\rho$ and let $x \in V_L$ be a leaf of $T$. Then we denote by $\delta_x$ the \emph{depth} of $x$ in $T$, which is the number of edges on the unique shortest path from $\rho$ to $x$. Then, the \emph{height} of $T$ is defined as $h(T) = \max\limits_{x \, \in  \, V_L} \delta_x$. 

Lastly, we want to introduce two particular trees which will be of interest later on, namely the \emph{caterpillar tree} $T_{n}^{cat}$, the unique rooted binary tree with $n$ leaves that has only one cherry, and the \emph{fully balanced tree} $T_{k}^{bal}$ with $n= 2^k$ leaves in which all leaves have depth precisely $k$ (Figure \ref{Fig_TbalTcat}).

\begin{figure}
	\centering
	\includegraphics[scale=0.35]{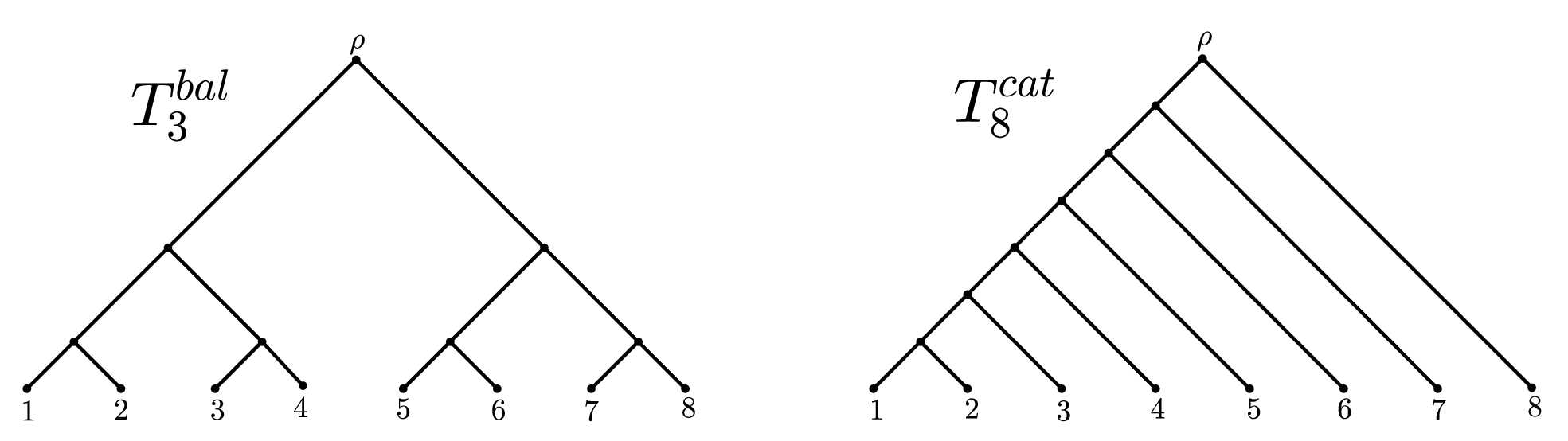}
	\caption{Fully balanced tree of height 3, $T_3^{bal}$, and caterpillar tree, $T_8^{cat}$, on 8 leaves.}
	\label{Fig_TbalTcat}
\end{figure}

Now that we have introduced the concept of a tree, we need to introduce the data that we will map onto the leaves of a tree, i.e. we need to introduce binary characters. A \emph{binary character} $f$ is a function $f: X \rightarrow \{0,1\}$ from the taxon set $X$ to the set $\{0,1\}$.
We often abbreviate a character $f$ by denoting it as $f= f(1)f(2) \ldots f(n)$ for $X=\{1, \ldots, n\}$. 
Moreover, given a binary character $f$ we use $\bar{f}$ to denote its inverted character (where we replace all ones by zeros and vice versa), e.g. if $f=0011$, then $\bar{f}=1100$.  
A binary character on state set $\mathcal{R} = \{0,1\}$ is called \emph{informative} if it contains at least two zeros and two ones. Recall that in biology, a sequence of characters is also often called an \emph{alignment} $\mathcal{A}$.
Furthermore, an \emph{extension} of a binary character $f$ is a map $g: V \rightarrow \{0,1\}$ such that $g(i)=f(i)$ for all $i \in X$.
Last but not least, for a phylogenetic tree $T$, we call $ch(g) = \vert \{ (u,v) \in E, \, g(u) \neq g(v)\} \vert$ the \emph{changing number} of $g$ on $T$. 

\subsection*{Persistence}
Let $T$ be a phylogenetic tree with root $\rho$ and an additional root edge, whose second endnode is $\rho'$. Throughout this manuscript we assume that the state of $\rho'$ is 0. Then we call a binary character \emph{persistent} if it can be realized on $T$ by at most one $0 \rightarrow 1$ transition followed by at most one $1 \rightarrow 0$ transition in the tree. 
More formally, we call a character $f$ persistent if there exists an extension $g$ of $f$ that realizes $f$ with at most one $0 \rightarrow 1$ transition and at most one $1 \rightarrow 0$ transition. 
We call such an extension a \emph{minimal persistent extension} if it minimizes the changing number $ch(g)$. 
As an example consider the caterpillar tree $T$ on four leaves depicted in Figure \ref{Fig_PersNonPers} and the two characters $f_1 = 0110$ and $f_2=0101$. $f_1$ is persistent, because it can be realized by a $0 \rightarrow 1$ change on edge $(\rho, u)$ followed by a $1 \rightarrow 0$ change on edge $(v,1)$. There is a unique minimal extension for $f_1$, namely the one that assigns state $0$ to $\rho$ and state $1$ to $u$ and $v$. $f_2$, however, is not persistent, because it would require at least two $0 \rightarrow 1$ transitions in the tree (or one $0 \rightarrow 1$ transition followed by two $1 \rightarrow 0$ transitions). It can easily be verified that there exists no extension of $f_2$ such that $f_2$ is persistent. 

In general, we denote by a tuple $(f,\mathcalligra{p}(f,T))$ a character $f$ together with its {\em persistence status $\mathcalligra{p}$ on a given tree $T$}, where $$\mathcalligra{p}(f,T)=\begin{cases}p & \mbox{ if $f$ is persistent on $T$} \\ np & \mbox{else,} \end{cases}$$ i.e. $\mathcalligra{p}$ is the persistence indicator function.

Additionally, for a character $f$ that is persistent on $T$, we define $l_p(f,T)$ as the {\em persistence score of $f$ on $T$}, i.e. 
$$l_p(f,T)=\begin{cases}0 & \mbox{if $f = 0\ldots 0$, i.e. $f$ requires no changes,} \\ 1 & \mbox{if $f$ requires one $0\rightarrow 1$ change,} \\ 2 & \mbox{if $f$ requires one $0\rightarrow 1$ change followed by one $1 \rightarrow 0$ change}. \end{cases}$$
Note that with `requires' we explicitly mean that $f$ cannot be realized with fewer changes. 

\begin{figure}[htbp]
	\centering
	\includegraphics[scale=0.25]{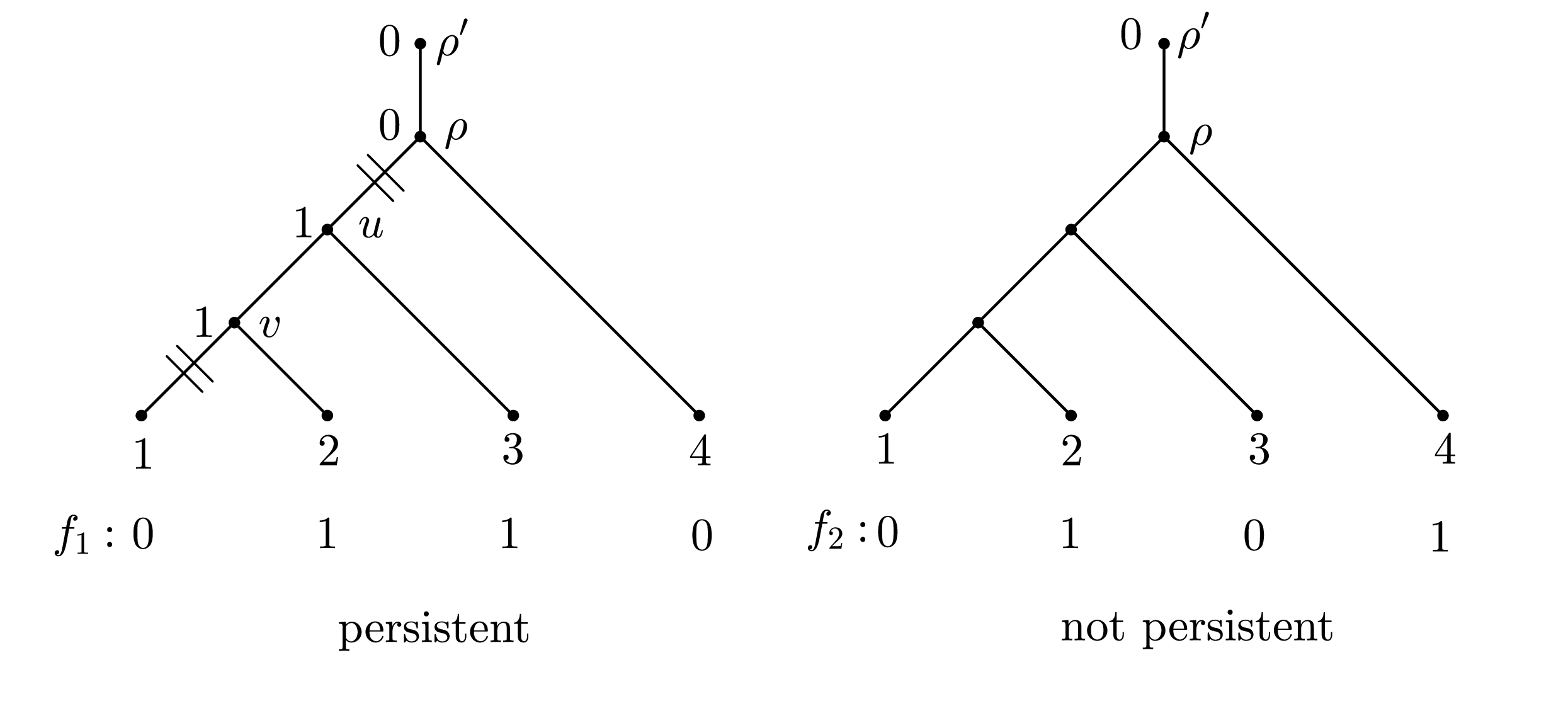}
	\caption{Caterpillar tree $T^{cat}_4$ with a persistent and a non-persistent character assigned to its leaves. Character $f_2$ on the right is not persistent, because there exists no persistent extension of $f_2$.}
	\label{Fig_PersNonPers}
\end{figure}

\subsection*{Maximum Parsimony and the Fitch Algorithm}
As we will later on establish a relationship between persistent characters and the so-called Fitch algorithm, we now introduce the criterion of Maximum Parsimony.

Given a character $f$, the idea of Maximum Parsimony is to find a phylogenetic tree $T$ that minimizes the \emph{parsimony score} $l(f,T)$ of $f$, where $l(f,T) = \min\limits_{g} ch(g,T)$ and the minimum runs over all possible extensions $g$ of $f$. For a given tree $T$, an extension $g$ that minimizes the changing number is called a \emph{most parsimonious extension} or {\em minimal extension} and a tree $T$ that minimizes the parsimony score is called a \emph{Maximum Parsimony tree}. Note that this concept is similar to the persistence score $l_p(f,T)$ introduced above and we will see later on how these two scores relate to each other.
Given a phylogenetic tree $T$ and a character $f$, we can use the so-called Fitch algorithm (\citet{Fitch}) in order to calculate the parsimony score $l(f,T)$ as well as a minimal extension $g$ of $f$ that realizes $f$ on $T$ with $l(f,T)$ changes. Formally, the Fitch algorithm consists of three phases, but we will only consider the first two phases here. The first phase, which is most important for our purposes, is based on \emph{Fitch's parsimony operation} $\ast$, which is defined as follows: Let $A, B \subseteq \mathcal{R}$. Then
$$ A \ast B = \begin{cases} 
				A \cap B, \, \text{ if } A \cap B \neq \emptyset \\
				A \cup B, \, \text{ otherwise.} 
			  \end{cases}$$
In principle, the first phase of the Fitch algorithm goes from the leaves to the root of a tree $T$ and assigns each parental node a state set based on the states of its children. First, each leaf is assigned the set consisting of the state assigned to it by $f$. Then all other nodes $v$, whose children both have already been assigned state sets, say $A$ and $B$, are assigned the set $A \ast B$. 
Note that the parsimony score $l(f,T)$ corresponds to the number of times the union is taken. 
Moreover, note that in this manuscript $\mathcal{R}=\{0,1\}$. Whenever a node is assigned state set $\{0,1\}$ as result of the union of $\{0\}$ and $\{1\}$ being taken, we call it a \emph{$\{0,1\}$ union node}. Else, if the assignment of state set $\{0,1\}$ results from the intersection of $\{0,1\}$ and $\{0,1\}$ being taken, we refer to it as a \emph{$\{0,1\}$ intersection node}. This distinction will be of relevance in subsequent analyses, because whenever there are only two $\{0,1\}$ nodes in a tree, we know that each of them must be a union node.
As an example, consider the caterpillar tree $T_4^{cat}$ depicted in Figure \ref{Fig_Fitch}, where this concept is illustrated.\\
The second phase of the Fitch algorithm goes from the root to the leaves in order to compute a minimal extension of $f$. First, the root is arbitrarily assigned one state of its state set that was assigned during the first phase of the algorithm. Then, for every inner node $v$ that is a child of a node $u$ that has already been assigned a state, say $g(u)$, we do the following: if $g(u)$ is contained in the ancestral state set of $v$, we set $g(v)=g(u)$. Otherwise, we arbitrarily assign any state from the ancestral state set of $v$ to $v$. Overall, this gives a minimum extension $g$ of $f$ and we will later on see how the Fitch algorithm relates to persistence. 

\begin{figure}[htbp]
	\centering
	\includegraphics[scale=0.2]{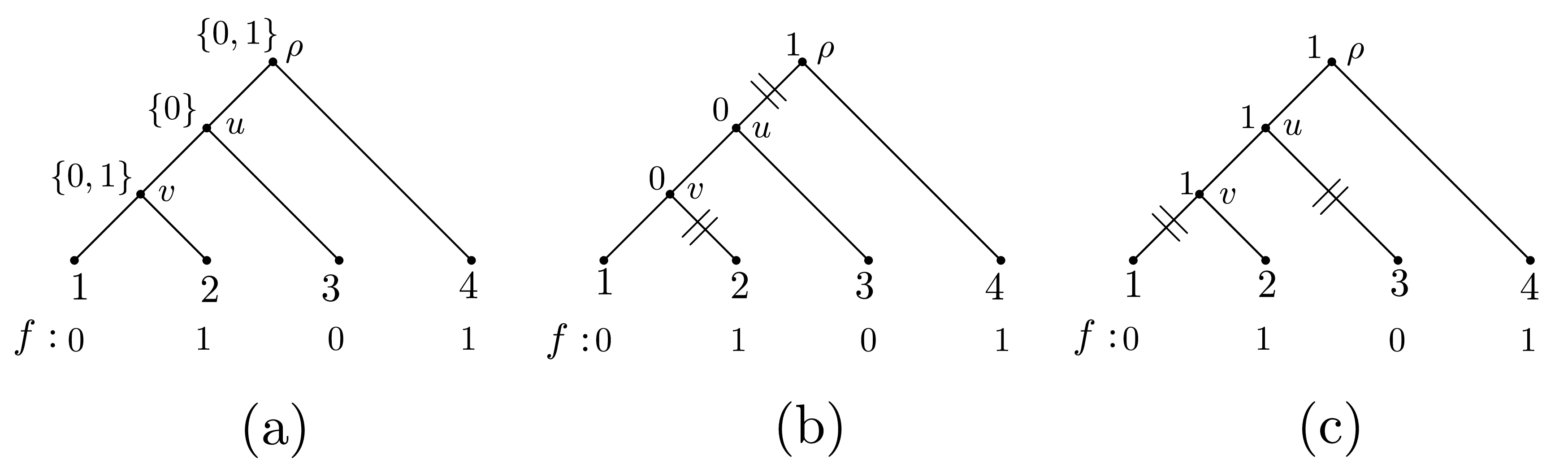}
	\caption{Caterpillar tree $T_4^{cat}$ together with character $f=0101$. In (a) the ancestral state sets found by the first phase of the Fitch algorithm are depicted. Note that $l(f,T)=2$ as there are precisely two $\{0,1\}$ union nodes. In (b) a most parsimonious extension found by the second phase of the Fitch phase is depicted, while in (c) a most parsimonious extension not found by the Fitch algorithm is shown.}
	\label{Fig_Fitch}
\end{figure}

\subsection*{Tree Balance}
Having introduced the notion of persistent characters above, we will later not only be concerned with the relationship between persistence and the Fitch algorithm, but also with the relationship between persistence and tree shape, in particular tree balance. There are different methods to measure the balance of a tree, one of them being the Sackin index (\citet{Sackin}). 
Recall that the Sackin index of a rooted binary phylogenetic tree $T$ is defined as $\mathcal{S}(T)=\sum\limits_{u \in \mathring{V}(T)}n_u$, where $n_u$ is the number of leaves of the subtree of $T$ rooted at $u$. This index is maximized by the caterpillar tree (\citet{Fischer2018}).
If for two trees $T_1$ and $T_2$ we have $\mathcal{S}(T_1) < \mathcal{S}(T_2)$, then $T_1$ is called {\em more balanced} than $T_2$.

As an example consider the caterpillar tree $T_4^{cat}$ depicted in Figure \ref{Fig_Fitch}. Here, we have $\mathcal{S}(T_4^{cat}) = 2+3+4=9$, because $v$ has two descendant leaves, $u$ has three descendant leaves and $\rho$ has four.

In the following we will see that the more balanced a tree, the fewer binary characters are persistent on it.

\subsection*{Splits and Compatibility}
One last concept that we need to introduce is that of splits and compatibility. In principle, splits are a way to describe unrooted phylogenetic trees. Given a set $X$, a bipartition of $X$ into two non-empty subsets $A$ and $B$ is called an $X$\emph{-split} and is denoted by $\sigma = A \vert B$. 
We call two $X$-splits $\sigma_1 = A_1 \vert B_1$ and $\sigma_2 = A_2 \vert B_2$ \emph{compatible}, if at least one of the following intersections is empty: $A_1 \cap A_2, \, A_1 \cap B_2, \, A_2 \cap B_1$, or $B_1 \cap B_2$ (cf. \citet{Semple2003}).

Now, let $T$ be an unrooted phylogenetic tree on $X$ and let $e$ be an edge of $T$. Then the removal of $e$ splits $T$ into two connected components, say $T_1$ and $T_2$. Let $X_i \subseteq X$ be the leaf set of $T_i$ for $i=1,2$. Then we say that edge $e$ \emph{induces} the split $\sigma_e = X_1 \vert X_2$. Note that every edge leading to a  leaf $x \in X$ (sometimes also called a \emph{pending edge}) induces a so-called \emph{trivial} split $\{x\} \vert X \setminus \{x\}$. We will denote the set of all induced \emph{non-trivial splits} (i.e. splits $\sigma_e = X_1 \vert X_2$ with $\vert X_1 \vert, \, \vert X_2 \vert \geq 2$) by $\Sigma^{\ast}(T)$.
Note that a non-trivial split $\sigma$ can be translated into an informative binary character $f_{\sigma}$ by assigning one of the two states to all taxa in $X_1$ and the other one to the taxa in $X_2$.

We will later on need these concepts when we consider the question of how many binary characters together with their persistence status are needed to uniquely determine a phylogenetic tree.

\section{Results}

We are now in the position to illustrate our results concerning the notion of persistence, where we start with elaborating the connection between persistence and Maximum Parsimony.

\subsection{Links between persistence and Maximum Parsimony}
The overall aim of this section is to fully characterize persistent characters in terms of their parsimony score and the first phase of the Fitch algorithm. Based on this, we provide an algorithm to decide whether a character is persistent or not. We will see that a character $f$ with $l(f,T) \leq 1$ is always persistent on $T$, while a character $f$ with $l(f,T) > 2$ is never persistent. Characters with parsimony score $2$, on the other hand, may or may not be persistent on a given tree $T$ (Figure \ref{Fig_PersNonPers}). However, we will show that the ones that are persistent can be characterized with the help of the first phase of the Fitch algorithm. In the following we will thus prove the main theorem of this section:

\begin{theorem}[Characterization of persistent characters] \label{characterize}
Let $f$ be a binary character on $\mathcal{R}=\{0,1\}$ and let $T$ be a phylogenetic tree. Then, we have:
\begin{enumerate}
\item If $l(f,T) > 2$, then $f$ is not persistent on $T$.
\item If $l(f,T) \leq 1$, then $f$ is persistent on $T$. 
\item If $l(f,T) = 2$, let the two $\{0,1\}$ union nodes found during the 1st phase of the Fitch algorithm be denoted by $u$ and $w$, respectively. Then, we have:
 $$ f \mbox{ is persistent on } T $$ 
 $$\Leftrightarrow $$
\hspace*{3cm} all of the following conditions hold: \\

\begin{enumerate} 
\item $u$ is an ancestor of $w$ or vice versa; wlog $u$ is the ancestor of $w$. 
\item The ancestral state sets found by the first phase of the Fitch algorithm fulfill the following conditions (Figure \ref{Fig_characterize}):
	\begin{itemize}
	\item all nodes that are descendants of the direct descendant $v$ of $u$ on the path to $w$, but not of $w$ are assigned state set $\{1\}$ (in particular, all nodes on the path from $v$ to $w$ (including $v$) are assigned state set $\{1\}$),
	\item all nodes that are not descendants of $v$ are assigned state set $\{0\}$.
	\end{itemize}
\end{enumerate}
\end{enumerate}
\end{theorem}

\begin{remark}\label{atleastoneinbetween} Concerning the first part of 3(b) in the theorem, note that if a character has parsimony score 2, the two $\{0,1\}$ union nodes, of which one is by Part 3(a) of Theorem \ref{characterize} an ancestor of the other one, cannot be directly adjacent -- i.e. there is no edge connecting the two nodes. More precisely, while one of the two nodes is an ancestor of the other one, it cannot be a {\em direct} ancestor. This is due to the fact that by the Fitch algorithm, the only way to get a direct ancestor of a $\{0,1\}$ node to also be assigned $\{0,1\}$ is to have {\em both} children assigned $\{0,1\}$, not only one (if the other child is for instance assigned $\{0\}$, then the parent would be assigned the intersection of $\{0\}$ and $\{0,1\}$, namely $\{0\}$, rather than $\{0,1\}$). However, the parent $\{0,1\}$ node would then be an intersection node and not a union node, but we are only considering $\{0,1\}$ union nodes. In total, this guarantees that there is at least one node $v$ between $u$ and $w$, which is a direct descendant of $u$. In particular, we have $v\neq w$. \\
Moreover, as by 3(a), $u$ is an ancestor of $w$, of the two maximal pending subtrees of $w$, one only has nodes that are assigned state set $\{0\}$ and the other only has nodes that are assigned state set $\{1\}$. Again, this is due to the fact that $u$ and $w$ are the \emph{only} $\{0,1\}$ union nodes.
\end{remark}

The proofs of Parts 1 and 2 of Theorem \ref{characterize} are relatively straightforward and can be found in the appendix. The proof of Part 3 requires several intermediate results. The general strategy is to show that $f$ has a minimal persistent extension on $T$ if and only if conditions (a) and (b) hold. If they hold, this minimal persistent extension is unique and can be found with the first phase of the Fitch algorithm. %\end{remark}

Before elaborating on this, we illustrate how the persistence score and the parsimony score of a character $f$ on a tree $T$ relate to each other and prove the following statements, which in turn will be used to prove Parts 1 and 2 of Theorem \ref{characterize}.

\begin{lemma} \label{MPboundsPersistence} Let $\mathcalligra{p}(f,T)=p$. Then,
\begin{enumerate}
\item $l_p(f,T) \geq l(f,T)$.
\item $l(f,T)\leq 2$.
\end{enumerate}
\end{lemma}

\begin{proof} \leavevmode
\begin{enumerate}
\item By definition, $l(f,T)$ denotes the minimum number of changes required to realize $f$ on $T$. Thus, the number of changes of any persistent extension is at least $l(f,T)$.
\item By part 1 and by the definition of the persistence score, we have $l(f,T)\leq l_p(f,T)\leq 2$.
\end{enumerate}
\end{proof}

\begin{remark}
Note that the first part of the lemma does not state equality for persistent characters, because for instance if $f=1,\ldots,1$, i.e. $f$ is the constant 1-character, we have $l(f,T)=0$ for any tree $T$, whereas we have $l_p(f,T)=1$, because a $0 \rightarrow 1$ change is needed on the root edge. 
\end{remark}

In the following we will establish a connection between the ancestral state sets found by the 1st Fitch phase and persistent extensions, in particular for characters with parsimony score 2. However, we start with characterizing minimal persistent extensions.

\begin{lemma} \label{notadjacent} Let $f$ be persistent on $T$ and let $g$ be a minimal persistent extension  of $f$ on $T$. Then, if $g$ contains a $0 \rightarrow 1$ change on an edge $(u,v)$ and a subsequent $1 \rightarrow 0$ change on an edge $(w,x)$, these two change edges are not adjacent, i.e. $v \neq w$. 
\end{lemma}

The proof of this lemma can be found in the appendix. Before we can proceed to characterize all persistent characters with parsimony score 2, we further characterize minimal persistent extensions of such characters in relation to the first phase of the Fitch algorithm.

\begin{lemma} \label{01nodes} Let $f$ be persistent on $T$ and $l(f,T)=2$ and let $g$ be a minimal persistent extension of $f$ on $T$. Then, $g$ has the property that all of its change edges have a source node that is assigned $\{0,1\}$ by the first phase of the Fitch algorithm. 
\end{lemma}

Again, the proof of this lemma can be found in the appendix. We will use this lemma subsequently, e.g. in the proof of Proposition \ref{unique}. However, note that the result of Lemma \ref{01nodes} does not necessarily hold if $l(f,T) \leq 1$, because then an extra change might be required on the root edge, which will not be captured by the Fitch algorithm. 

Next, recall that the first phase of the Fitch algorithm, which assigns possible ancestral states to internal nodes, does not necessarily find all possible ancestral states for each node.\footnote{Moreover, it might find too many states, which will never be used in the subsequent phases of the Fitch algorithm, but this is not relevant here (\citet{Fitch, Felsenstein2004}).} For example, if you consider the caterpillar tree $T_4^{cat}$ depicted in Figure \ref{Fig_Fitch} together with the character $f=0101$, the first Fitch phase returns ancestral state sets $\{0,1\}$ for node $v$, $\{0\}$ for node $u$ and $\{0,1\}$ for node $\rho$, respectively. In particular, the parent node of leaf $3$ is assigned state set $\{0\}$. So these sets would not support the choice of assigning 1 to all ancestral nodes -- but this choice would also lead to a changing number of 2 and thus be a most parsimonious extension. This is why the Fitch algorithm requires a correction phase if {\em all} most parsimonious extensions are needed (\citet{Fitch}).

However, while this example shows that the first phase of the Fitch algorithm might miss some most parsimonious extensions for $f$ on $T$ even if $f$ is binary, we will now show that -- if $f$ is persistent -- the first phase of the Fitch algorithm always suffices to reconstruct a minimal persistent extension. In particular, it cannot happen that we miss all minimal persistent extensions when using the first phase of the Fitch algorithm. Moreover, if a minimal persistent extension exists, it is unique. So compared to the fact that there might be many most parsimonious extensions  -- some of which cannot even be found by the first phase of the Fitch algorithm -- the following proposition is remarkable.

\begin{proposition} \label{unique} 
Let $f$ be persistent on $T$. Then, $f$ has a unique minimal persistent extension on $T$.
\end{proposition}

\begin{proof}
By Lemma \ref{MPboundsPersistence}, Part 2, we need to distinguish three cases. 
\begin{enumerate}
\item If $l(f,T)=0$, we know that $f=0,\ldots,0$ or $f=1,\ldots,1$. In the first case, assigning state 0 to all nodes of $T$ leads to an extension that requires no changes. This must be minimal, and it is also persistent. As there is no other extension requiring no changes, this assignment is also unique. On the other hand, if $f=1,\ldots,1$, assigning state 1 to all nodes of $T$ -- except for $\rho'$, which by definition must be assigned 0 -- will lead to an extension with one $0 \rightarrow 1$ change on the root edge but no changes otherwise. Note that as $f$ employs 1 as a state, there must be a $0 \rightarrow 1$ change somewhere in $T$, which is why one change is already best possible. Moreover, as there is no other extension requiring only one change, this assignment is unique.
\item If $l(f,T)=1$, \red{we can conclude that there is a unique most parsimonious extension of $f$ on $T$. To see this, note that $l(f,T)=1$ implies that $f$ requires a change on precisely one edge of $T$.} In particular, this edge subdivides $T$ into two subtrees, one of which only contains leaves in state 0, whereas the other one contains only leaves in state 1. If we assign the respective state to all inner nodes, too, we end up with one subtree with only nodes in state 0 and the other subtree with only nodes in state 1. This gives us the unique most parsimonious extension in this case. We now have to show that it is also persistent.

Note that as  $l(f,T)=1$, $f$ cannot be constant, so we know that $f$ employs state 1 and thus at least one change is necessary. So if the change on the inner edge corresponding to $f$ is a $0 \rightarrow 1$ change, we are done -- our extension is persistent with $l_p(f,T)=1$. If, on the other hand, the change is a $1 \rightarrow 0$ change, we need to add a $0 \rightarrow 1$ change to the root edge in order to make the extension persistent, which leads  to $l_p(f,T)=2$. Note that there is no other edge where we could add this change, as any other choice would require additional changes, which are not permitted. So in any case, the extension we found is persistent, minimal and unique.

\item If $l(f,T)=2$ and $f$ is persistent, we know by Lemma \ref{01nodes} that the first phase of the Fitch algorithm assigns $\{0,1\}$ to two nodes which -- by definition of persistence --  lie on one path from the root to some leaf (i.e. one of them is an ancestor of the other one), and we know that any minimal persistent extension of $f$ must have a $0 \rightarrow 1$ change starting at the $\{0,1\}$ node closer to the root, say $u$, and a $1 \rightarrow 0$ change starting at the $\{0,1\}$ node further away from the root, say $w$. Moreover, we know by Lemma \ref{notadjacent} that the two $\{0,1\}$ nodes are not adjacent. As $l(f,T)=2$, we can thus conclude that these two are the only two $\{0,1\}$ nodes (because the scenario that one $\{0,1\}$ node is the parent of two other $\{0,1\}$ nodes cannot happen).

Now let us consider $w$ first. We know that $w$ is assigned $\{0,1\}$ and that $w$ is the source of the $1 \rightarrow 0$ edge, say $(w,x)$. So $w$ must be in state 1 and $x$ in state 0 in any minimal persistent extension. Analogously, $u$ as the source of the $0 \rightarrow 1$ edge, say $(u,v)$, must therefore be in state 0 and $v$ in state 1. 

As $f$ is persistent, we know that any persistent extension $g$ will be such that all nodes descending from $x$ must be in state 0 (because after the $1 \rightarrow 0$ change, no more changes are possible), and all nodes descending from $v$ but not from $x$ must be in state 1. Moreover, all nodes of $T$ that are not descendants of $v$ must be in state 0. So altogether, this induces only one minimal persistent extension, because there is no freedom to make alternative choices, which completes the proof.
\end{enumerate} 
\end{proof}

We have now analyzed minimal persistent extensions in terms of the first phase of the Fitch algorithm and have seen that if a character is persistent, it has a unique minimal persistent extension. Now, we are in the position to fully characterize all persistent characters and prove the main result of this section, namely Theorem \ref{characterize}. We only show the proof for Part 3 of Theorem \ref{characterize}. The proofs of Parts 1 and 2 are relatively straightforward and can be found in the appendix.

\begin{proof}[Proof of Theorem \ref{characterize}, Part 3]  \leavevmode
Let $l(f,T) = 2$. First we assume that $f$ is persistent and show that then conditions (a) and (b) hold. As $l(f,T)=2$ and $f$ is persistent on $T$, by Lemma \ref{MPboundsPersistence} and the definition of persistence, we have $l_p(f,T)=2$. By Proposition \ref{unique}, $f$ has a unique minimal persistent extension on $T$. This is in fact the only persistent extension, because as $l_p(f,T)=2$, we require at least two changes, which is why no extension with more changes can be persistent. We now show properties (a) and (b).
	
	\begin{enumerate}
	\item [(a)] We first need to show that one of the two $\{0,1\}$ union nodes found during the first phase of the Fitch algorithm is an ancestor of the other $\{0,1\}$ union node. As in the last part of the proof of Proposition \ref{unique}, this follows from Lemma \ref{01nodes}, because as $l(f,T)=2$, we know that any extension of $f$ requires at least two changes, and as $f$ is persistent, we know that this can be achieved by a $0 \rightarrow 1$ change and a subsequent $1 \rightarrow 0$ change. So such a persistent extension is automatically minimal and therefore, by Lemma \ref{01nodes}, the source nodes of both change edges correspond to the two $\{0,1\}$ union nodes of the 1st phase of the Fitch algorithm. Therefore, as the $1 \rightarrow 0$ change can by definition only affect descendants of the $0 \rightarrow 1$ change, one of the two  $\{0,1\}$ nodes must be an ancestor of the other one.

	\item [(b)] Now we want to show that the ancestral state sets found by the first phase of the Fitch algorithm are such that all nodes that are descendants of $v$ but not of $w$ are assigned state set $\{1\}$ and all nodes that are not descendants of $v$ are assigned state set $\{0\}$ (Figure \ref{Fig_characterize}). As we assume that there are exactly two $\{0,1\}$ union nodes, namely $u$ and $w$, and have shown that one is an ancestor of the other (in our case, without loss of generality, $u$ is an ancestor of $w$; note, however, that $u$ cannot be a direct ancestor of $w$ due to Lemma \ref{notadjacent}, which directly implies that there cannot be a third $\{0,1\}$ node), by Lemma \ref{01nodes} we know that the $0 \rightarrow 1$ change must have source node $u$ and the $1 \rightarrow 0$ change must have source node $w$. Moreover, as node $v$ lies on the path from $u$ to $w$, all nodes that are descendants of $v$ but not of $w$ have to be in state 1 (and will thus be assigned state set $\{1\}$), because they are affected by the $0 \rightarrow 1$ change starting on the edge $(u,v)$, but not by the $1 \rightarrow 0$ change starting in $w$. On the other hand, all nodes that are not descendants of $v$ and are thus not affected by the $0 \rightarrow 1$ change, have to be in state 0 and are  assigned state set $\{0\}$ by the first phase of the Fitch algorithm.
	\end{enumerate}
	
So if $f$ is persistent, then all of the conditions hold. 	

We now assume that both conditions hold and prove that this is sufficient for $f$ to be persistent:

If one of the two union nodes that are assigned $\{0,1\}$ by the first Fitch phase is the ancestor of the other one (with at least one node in-between), then we can use Lemma \ref{01nodes} to conclude that the two change edges must start at these nodes. However, this alone would not be sufficient, because for persistence, the changes have to occur in the right order. However, as we also assume that the ancestral state sets found by the first phase of the Fitch algorithm are as claimed in condition (b), we can directly construct a minimal persistent extension. For all nodes for which the first phase of the Fitch algorithm makes an unambiguous choice, i.e. $\{0\}$ or $\{1\}$, we assign the corresponding state to the respective node. 
In particular, in one of the two maximal pending subtrees of $w$ we assign state $0$ to all nodes, while we assign state $1$ to all nodes in the other maximal pending subtree of $w$.
Moreover, we set $g(u)=0$ and $g(w)=1$. This gives us a persistent extension, requiring exactly two changes: a $0 \rightarrow 1$ change on the edge $(u,v)$ and a $1 \rightarrow 0$ change on an edge $(w,x)$ with source $w$. Thus, $l_p(f,T)=2$ and we can conclude that $f$ is persistent on $T$. This completes the proof.
\end{proof}

\begin{figure}[htbp]
	\centering
	\includegraphics[scale=0.1]{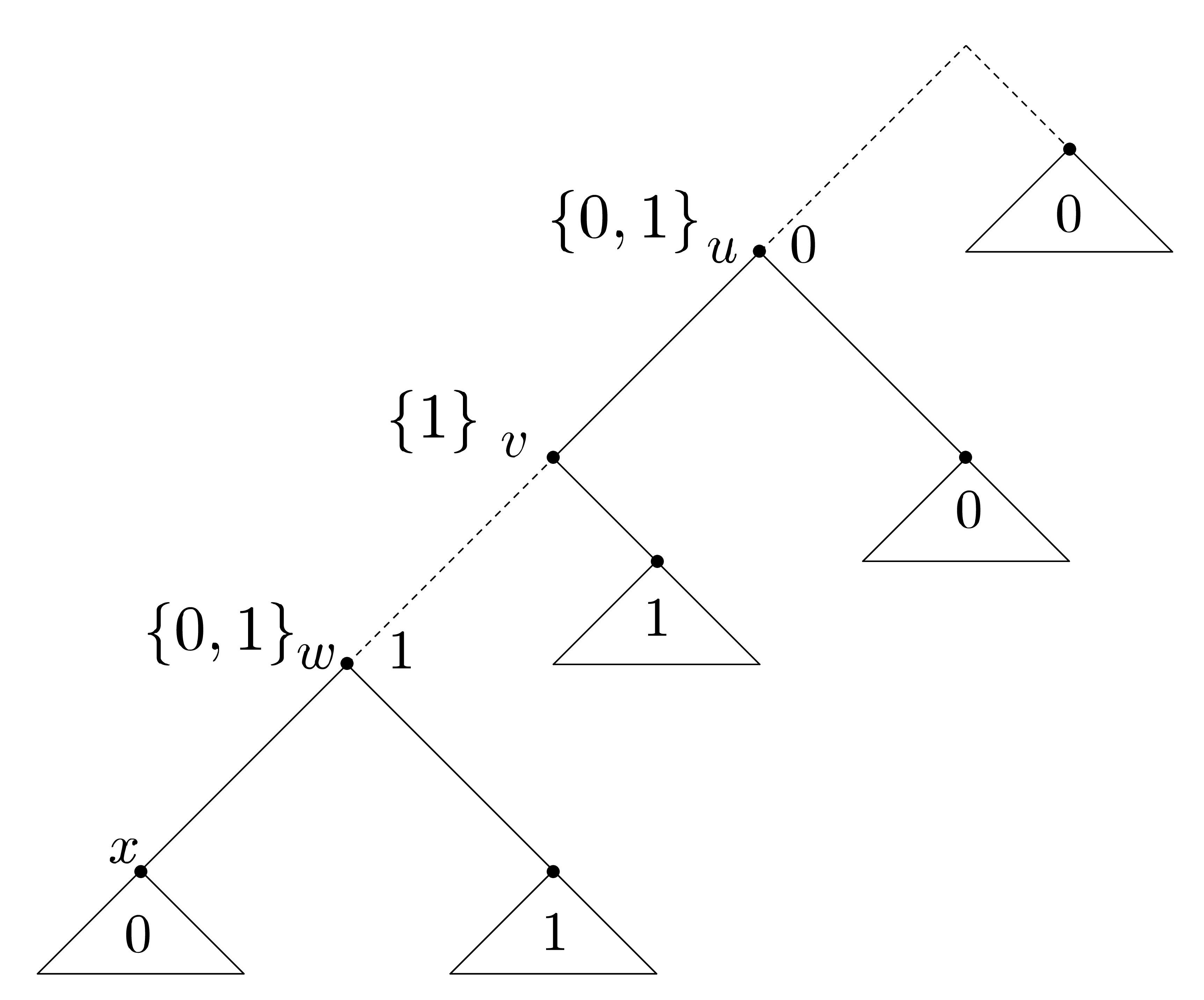}
	\caption{Conditions of Part 3 of Theorem \ref{characterize}. There are two nodes assigned state set $\{0,1\}$ by the first phase of the Fitch algorithm, namely $u$ and $w$, where $u$ is an ancestor of $w$, but $u$ and $w$ are not adjacent. Of the two maximal pending subtrees of $w$, one only has nodes that are assigned state set $\{0\}$ and the other only has nodes that are assigned state set $\{1\}$. Moreover, all nodes that are descendants of $v$ but not of $w$ are assigned set $\{1\}$, and all nodes that are not descendants of $v$ are assigned state set $\{0\}$.}
	\label{Fig_characterize}
\end{figure}

Note that if Condition (a) of Part 3 in Theorem \ref{characterize} holds and we know that $l(f,T)=2$, then the conditions in (b) imply that there is a most parsimonious extension $g$ of $f$ coinciding with the minimal persistent extension that we have constructed in the proof of Part 3 of Theorem \ref{characterize}. To be precise, we have the following corollary:
%%%%
\begin{corollary} \label{FitchSetsMostPars}
If we have $l(f,T)=2$ and Conditions (a) and (b) of Part 3 in Theorem \ref{characterize} hold, then there is a most parsimonious extension $g$ of $f$ such that
	\begin{itemize}
	\item $u$ is assigned state 0 and $w$ is assigned state 1,
	\item of the two maximal pending subtrees of $w$, one only has nodes in state 0 and the other one only has nodes in state 1,
	\item all nodes that are descendants of $v$ but not of $w$ are assigned state 1 (in particular, all nodes on the path from $v$ to $w$ are assigned state 1),
	\item all nodes that are not descendants of $v$ are assigned state 0.
	\end{itemize}
\end{corollary}
\begin{proof} The corollary is a direct consequence of the construction given in the end of the proof of Theorem \ref{characterize}.
\end{proof}

Thus, even though there might be other most parsimonious extensions, $g$ as described in Corollary \ref{FitchSetsMostPars} is one of them and coincides with the unique minimal persistent extension of $f$ constructed in the proof of Part 3 of Theorem \ref{characterize}. Thus, there might be several most parsimonious extensions of $f$, but only one of them (namely $g$) is also a minimal persistent one. \\

Summarizing the above, we have seen that we can calculate the persistence status of a character $f$ on tree $T$ solely based on the first phase of the Fitch algorithm, because we only require the ancestral state sets found by the first phase of the Fitch algorithm (which in turn give us the parsimony score of $f$ on $T$). Note that the crucial idea of deciding whether a character $f$ with parsimony score 2 is persistent or not used in Theorem \ref{characterize} was to consider the distribution of the ancestral state sets $\{0,1\}, \, \{0\}$ and $\{1\}$ across the tree. In particular one of the two  $\{0,1\}$ union nodes must be an ancestor of the other $\{0,1\}$ union node, while all nodes on the path between them (of which at least one must exist) are $\{1\}$ nodes. This leads to an algorithm which can be found in the appendix (Algorithm \ref{CheckPersistence}).

The distribution of the ancestral state sets $\{0,1\}, \, \{0\}$ and $\{1\}$ across the tree, will also help us in the following section, where we will be concerned with counting the number of characters that are persistent on a given tree $T$. The idea of using the set assignments by the first Fitch phase is particularly helpful when counting the number of persistent characters with parsimony score 2.

\subsection{On the impact of the tree shape on the number of persistent characters}\label{sec_shape}
We will now turn to the relationship between the shape of a tree and its persistent characters. 
The main aim of this section is to show that the more imbalanced a tree is, the more persistent characters it has.
In particular, we will show that there is a direct relationship between the so-called Sackin index of a phylogenetic tree (\citet{Sackin}) and its number of persistent characters. This is a surprising combinatorial result. It is particularly interesting and might be of practical impact as the Sackin index of a phylogenetic tree is very easy to calculate, while the number of persistent characters of a tree is not as straightforward to see per se.

In the following, we denote by $\mathcal{P}_i(T)$ the number of persistent characters of $T$ with parsimony score $i$ and by $\mathcal{P}(T)$ the number of all binary characters that are persistent on $T$. 

First note that, given a rooted binary phylogenetic tree $T$ with $n$ leaves, by Part 2 of Theorem \ref{characterize}, all characters with parsimony score at most 1 are persistent on $T$. So this gives us the two constant characters $0,\ldots,0$ and $1,\ldots,1$, which both have parsimony score 0, i.e. $\mathcal{P}_0(T)=2$. Moreover, considering parsimony score 1 gives us two times the number of characters corresponding to the $n-3$ non-trivial splits induced by \red{the inner edges of} $T$\footnote{\red{Note that $T$ has $n-2$ inner edges but the two edges incident to the root induce the same split; so in total, the inner edges of $T$ induce $n-2-1=n-3$ non-trivial splits.}} (because for each such split we have to consider $f$ and $\bar{f}$) and two times the number of characters corresponding to the $n$ trivial splits induced by $T$ (because, again, we have to consider $f$ and $\bar{f}$). So this implies $\mathcal{P}_1(T)= 2(n-3)+2n = 4n-6$. In total, we have $\mathcal{P}_0(T)+\mathcal{P}_1(T)=2+4n-6 = 4n-4$ persistent characters on $T$ -- regardless of any specific properties of $T$ like e.g. its tree shape. 

Moreover, by the second part of Lemma \ref{MPboundsPersistence}, we know that if $l(f,T)\geq 3$, $f$ cannot be persistent on $T$, which means $\mathcal{P}_i(T)=0$ for all $i\geq 3$. So in order to count all persistent characters on a given tree $T$, only the ones with parsimony score 2 are still missing. However, their number depends on the tree shape of $T$ as can be seen when considering the following example.

\begin{example}
Consider the two trees $T_1$ and $T_2$ depicted in Figure \ref{Fig_TreesP2}. Using the results of the previous section, in particular Proposition \ref{unique} and Theorem \ref{characterize}, we know that each persistent character with parsimony score 2 has a unique minimal persistent extension that has the property that the first phase of the Fitch algorithm will assign the two $\{0,1\}$ union sets such that one is the ancestor of the other one, but the two nodes may not be adjacent. For $T_1$ there is one such choice, leading to the two characters $f_1=1010$ and $f_2=0110$. Note that the two choices result from the fact that the roles of the two subtrees pending on the lowermost $\{0,1\}$ node can be interchanged. Now for $T_2$, there is no such choice, because there is no path from the root to any leaf employing at least three inner nodes. So $T_2$ has no persistent character with parsimony score 2. 
\end{example}

\begin{figure}[htbp]
	\centering
	\includegraphics[scale=0.1]{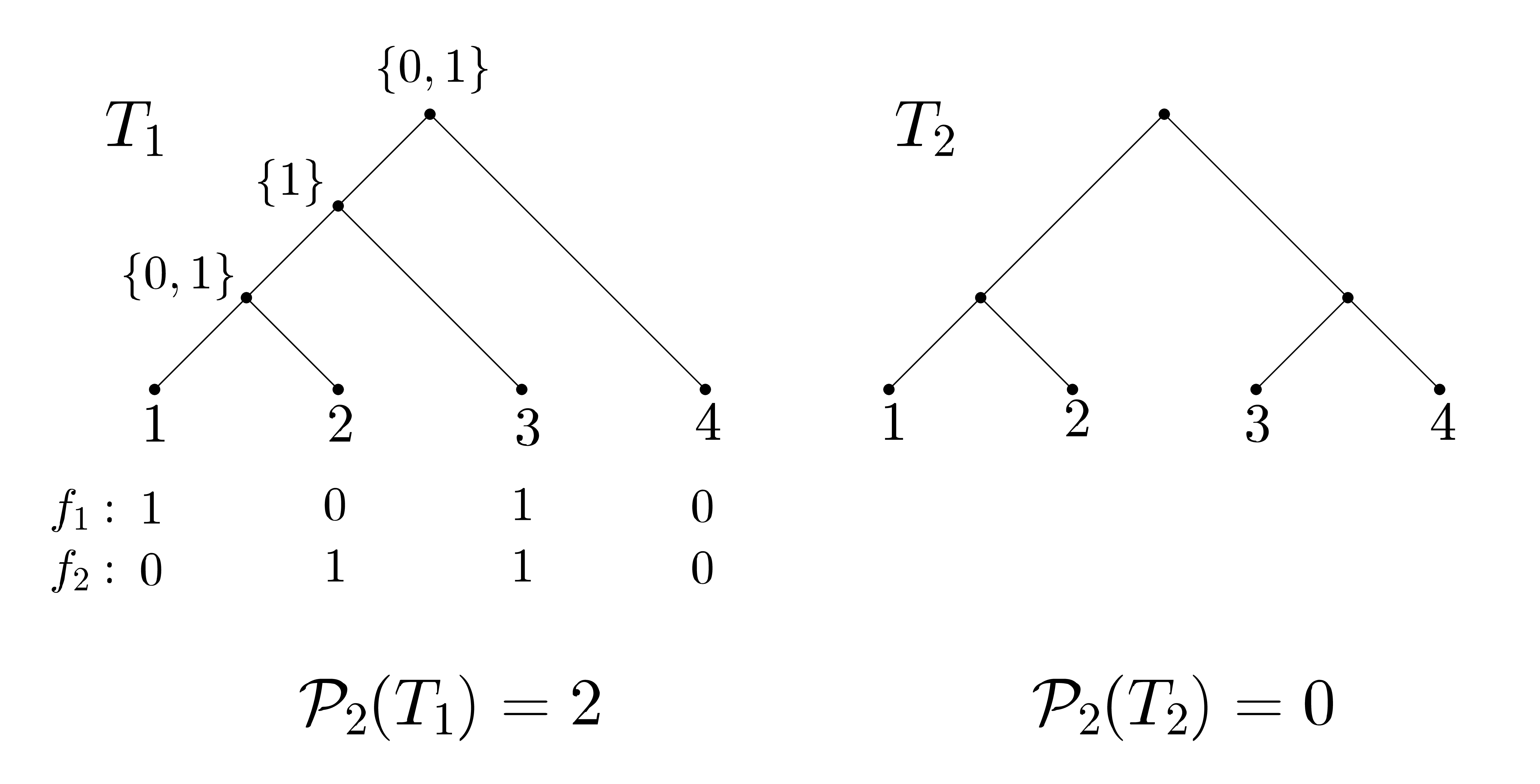}
	\caption{Trees $T_1$ and $T_2$ with $\mathcal{P}_2(T_1)=2$ and $\mathcal{P}_2(T_2) = 0$.}
	\label{Fig_TreesP2}
\end{figure}

Now, recall that $\mathcal{S}(T)$ denotes the so-called Sackin index of $T$, which can be used to evaluate the balance of a tree. 
In the previous example, tree $T_1$, for which we have $\mathcal{S}(T_1)=2+3+4=9$, has {\em more} persistent characters than $T_2$, for which we have $\mathcal{S}(T_2)=2+2+4=8$, so $T_1$ with more persistent characters also has a higher Sackin index. In the following, we will generalize this result and show that indeed no other tree shape has as many persistent characters as the caterpillar tree, which is the unique maximizer of the Sackin index (\citet{Fischer2018}). More importantly, we will show that, if the trees are ordered according to their Sackin index from balanced to very imbalanced (caterpillar), then the number of persistent characters per tree also increases. In other words, the more imbalanced a tree is, the more persistent characters it has. Proving the following theorem, which relates the number of persistent characters to the Sackin index, is thus the main aim of the present section.

\begin{theorem} \label{P2}
Let $T$ be a rooted binary phylogenetic tree with $n \geq 2$ leaves and root $\rho$. Then, the number $\mathcal{P}_2(T)$ of characters $f$ such that $l(f,T)=2$ and $\mathcalligra{p}(f,T)=p$ can be calculated as follows: 
$$\mathcal{P}_2(T)= 2\mathcal{S}(T)-6n+8.$$

Moreover, the total number $\mathcal{P}(T)$ of persistent characters for $T$ can be calculated as: 
$$ \mathcal{P}(T) = 2\mathcal{S}(T) - 2n +4.$$  
\end{theorem}

Before we can prove Theorem \ref{P2}, we require two lemmas. First, Lemma \ref{MinPersAlsoMostPars} shows that a minimal persistent extension is always also a most parsimonious extension when the root edge is ignored. 

\begin{lemma} \label{MinPersAlsoMostPars} Let $f$ be persistent on $T$. Let $g$ be a minimal persistent extension of $f$ on $T$ (including $\rho'$). Then, $g$ is also a most parsimonious extension on $T$ (excluding $\rho'$).
\end{lemma}

\begin{proof} If $l_p(f,T)=0$, then by definition, $f=0,\ldots,0$, and thus $l(f,T)=0$. The only persistent extension requiring zero changes on $T$ is then the one assigning state 0 to all internal nodes of $T$. This is at the same time a most parsimonious extension, so we are done.

Now, if $l_p(f,T)=1$, by definition $f$ requires one $0 \rightarrow 1$ change, so there are two cases: either $f = 1,\ldots,1$, i.e. $f$ is constant, or $f$ is not constant. In the first case, as $f$ employs state 1, one change is unavoidable for persistence. But the only minimal persistent extension assigns a $0 \rightarrow 1$ change to the root edge, i.e. all inner nodes of $T$ (excluding $\rho'$) are assigned state 1. This assignment is also most parsimonious, because on $T$ (excluding $\rho'$) this would lead to no changes, which corresponds to the parsimony score of $f=1,\ldots,1$ on $T$, which is $l(f,T)=0$. 

On the other hand, if $l_p(f,T)=1$ but $f$ is not constant, we have $l(f,T)\geq 1$. By Lemma \ref{MPboundsPersistence}, we know that $l(f,T)\leq l_p(f,T)=1$. So, altogether, $l(f,T)=1$. So any extension realizing one $0 \rightarrow 1$ change on $T$ will also be most parsimonious, as there cannot be any extension with fewer changes.

Last, if $l_p(f,T)=2$, any persistent extension by definition requires a $0 \rightarrow 1$ change followed by a $1 \rightarrow 0$ change. In particular, this implies that $f$ cannot be constant (the two constant characters have persistence scores of 0 and 1, respectively), and thus $l(f,T)\geq 1$. Given a minimal persistent extension $g$ of $f$, we now distinguish two cases: either the extension requires the $0 \rightarrow 1$ change on the root edge or not. If the $0 \rightarrow 1$ change happens on the root edge, there is only one change in $T$ (excluding $\rho')$. So this extension must be most parsimonious as $l(f,T)\geq 1$.

On the other hand, if both change edges, say $(u,v)$ and $(w,x)$, are in $T$, we cannot have $l(f,T)=1$, because we know that $f$ on $T$ then looks as follows: $u$ has an incident edge not on the path to $w$, whose descending leaves are all in state 0, and both $v$ and $w$ have an incident edge whose descending leaves are all in state 1, but $x$ has only descending leaves in state 0 (Figure \ref{Fig_MinPersAlsoMostPars}). So this character $f$ does not correspond to an edge of $T$, and thus we have $l(f,T)\geq 2$. Then, as $l_p(f,T)=2$ and using Lemma \ref{MPboundsPersistence}, we conclude $l(f,T)=2$, which in turn implies that the given persistent extension $g$ is most parsimonious. This completes the proof.
\end{proof}

\begin{figure}[htbp]
	\centering
	\includegraphics[scale=0.1]{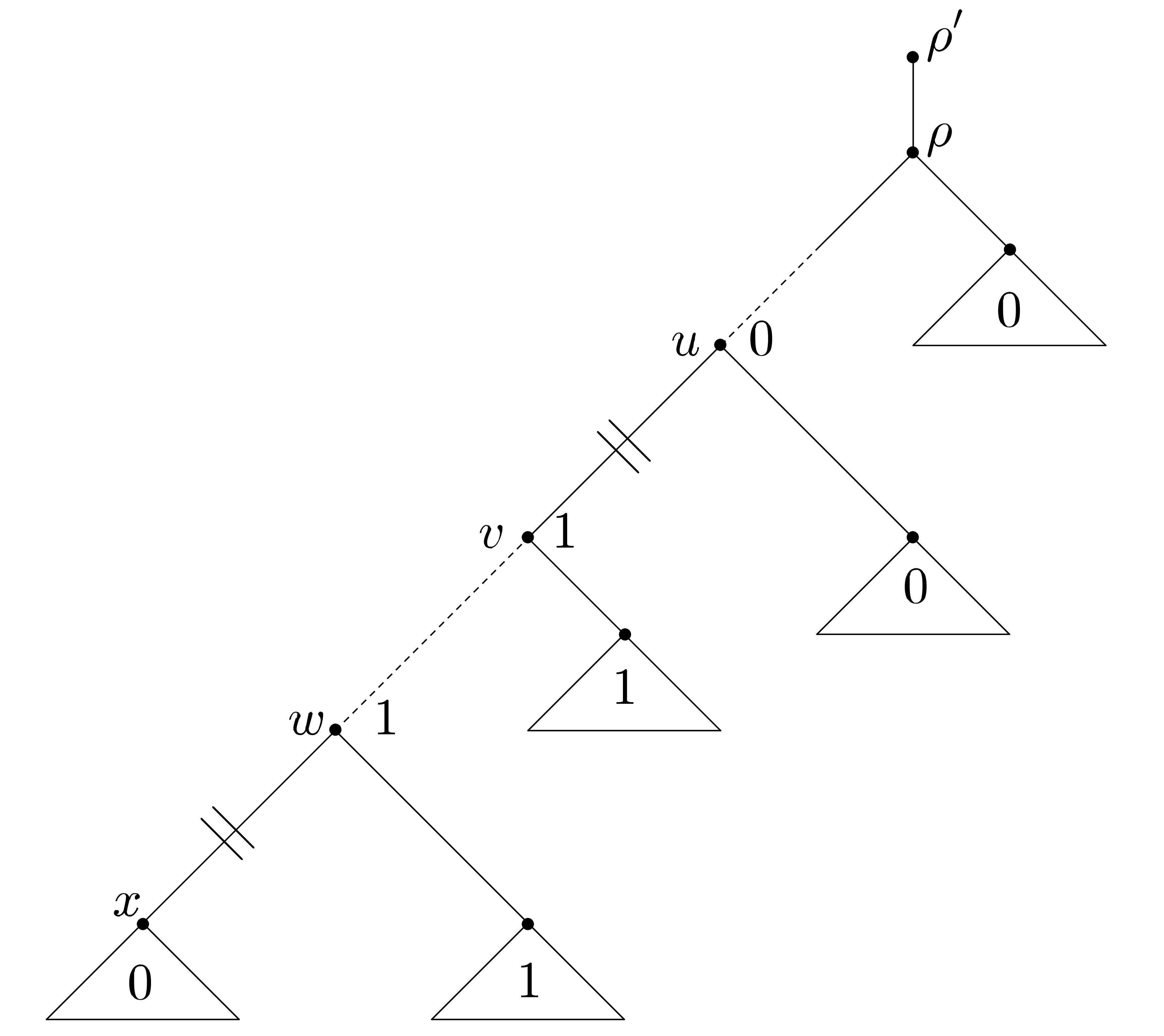}
	\caption{Situation in Lemma \ref{MinPersAlsoMostPars} in the case $l_p(f,T)=2$, where both change edges, namely $(u,v)$ and $(w,x)$ are in $T$ (i.e. there is no change on the edge $(\rho', \rho)$). This implies that all leaves descending from $v$ but not from $w$ are in state 1. Moreover, all leaves descending from $x$ are in state 0 and all leaves that are not descending from $v$ are also in state 0. Thus, $f$ does not correspond to an edge of $T$, and thus $l(f,T) \geq 2$.}
	\label{Fig_MinPersAlsoMostPars}
\end{figure}

Next, Lemma \ref{01induces2} shows that every choice of two $\{0,1\}$ union nodes for characters $f$ with $l(f,T)= 2$ leads to precisely 2 persistent characters.

\begin{lemma} \label{01induces2} Let $T$ be a rooted binary phylogenetic tree with $n\geq 4$ leaves. Let $u$ and $w$ be inner nodes of $T$ such that $u$ and $w$ are not adjacent and $u$ is an ancestor of $w$. Then, there are exactly two characters $f_1$ and $f_2$ that fulfill all of the following properties: 

\begin{enumerate} 
\item $f_1$ and $f_2$ are persistent on $T$, 
\item $l(f_1,T)=l(f_2,T)=2$,
\item the two $\{0,1\}$ union sets assigned to inner nodes by the first phase of the Fitch algorithm when evaluating $f_i$ on $T$ are $u$ and $w$ for $i=1,2$.
\end{enumerate}
\end{lemma}

\begin{proof} 
As $u$ and $w$ are inner nodes (i.e. they both have two direct descendants) which are not adjacent, one direct descendant of $u$, say $v$, must lie on the path between $u$ and $w$. Let $v'$ denote the other direct descendant of $u$ and let $x$ and $x'$ denote the direct descendants of $w$ (Figure \ref{Fig_01induces2}).
%Moreover, as $u$ and $w$ are inner nodes, they both have two direct descendants. 
We first consider $w$ and its direct descendants $x$ and $x'$. Note that as $u$ and $w$ are the only nodes assigned $\{0,1\}$ by the first phase of the Fitch algorithm (as there are only two $\{0,1\}$ union nodes, where one is an ancestor of the other, but not a direct one, there cannot be a third $\{0,1\}$ node in $T$), it is clear that all leaves descending from $x$ must be in state 0 and all leaves descending from $x'$ must be in state 1 or vice versa. This gives two options which we will call the $w$ options. 

Moreover, the direct ancestor of $w$, say $a$ (which might equal $v$), is not assigned $\{0,1\}$, so it must be assigned either $\{0\}$ or $\{1\}$ by the first Fitch phase. This leads to two options which we will call the $a$ options. But note that all leaves descending from $v$ (and thus also the ones descending from $a$) which are not also descendants of $w$ must be in the same state that is assigned to $a$. This is due to the fact that there is no $\{0,1\}$ node other than $w$ in the subtree rooted at $v$. So in particular, $v$ will be assigned the same state as $a$ by the first phase of the Fitch algorithm.

So in order for $u$ to be assigned $\{0,1\}$, its other child, say $v'$, must be assigned the set consisting of the state that is not in the set assigned to $a$. In particular, if $a$ (and thus $v$) is assigned $\{0\}$, then $v'$ must be assigned $\{1\}$ or vice versa. Again, as there is no other $\{0,1\}$ set in the tree, all leaves descending from $v'$ must be in the state assigned to $v'$. 

Moreover, all leaves that are not descending from $u$ can either all be in state 0 or all be in state 1 -- but there cannot be both states present as otherwise there would be an additional $\{0,1\}$ node. So this, again, gives rise to two options, which we will refer to as the $\rho$ options (as these leaves basically fix the state that will be assigned to $\rho$ by any most parsimonious extension).

So in total, combining the $w$ options with the $a$ options and $\rho$ options, we derive $2 \cdot 2 \cdot 2 = 8$ characters which have parsimony score 2 on $T$ and whose $\{0,1\}$ sets are assigned precisely to $u$ and $w$. These characters are illustrated by Figure \ref{Fig_01induces2}. As all minimal persistent extensions are by Lemma \ref{MinPersAlsoMostPars} also most parsimonious, only these characters are possible candidates for the persistent characters fulfilling the required properties. However, by considering Part 3 of Theorem \ref{characterize}, we conclude that only two of these eight characters can be persistent, because only the $w$ options give possible choices; the other options are fixed. In particular, $a$ and thus $v$ have to be assigned $\{1\}$ and all nodes that are not descendants of $v$ must be assigned $\{0\}$, so the $a$ and $\rho$ options leave only one choice (this is due to the fact that in a persistent extension, the first change has to be a $0 \rightarrow 1$ change and the $1 \rightarrow 0$ change can only follow afterwards). This concludes the proof. 
\end{proof}

\begin{figure}[htbp]
	\centering
	\includegraphics[scale=0.1]{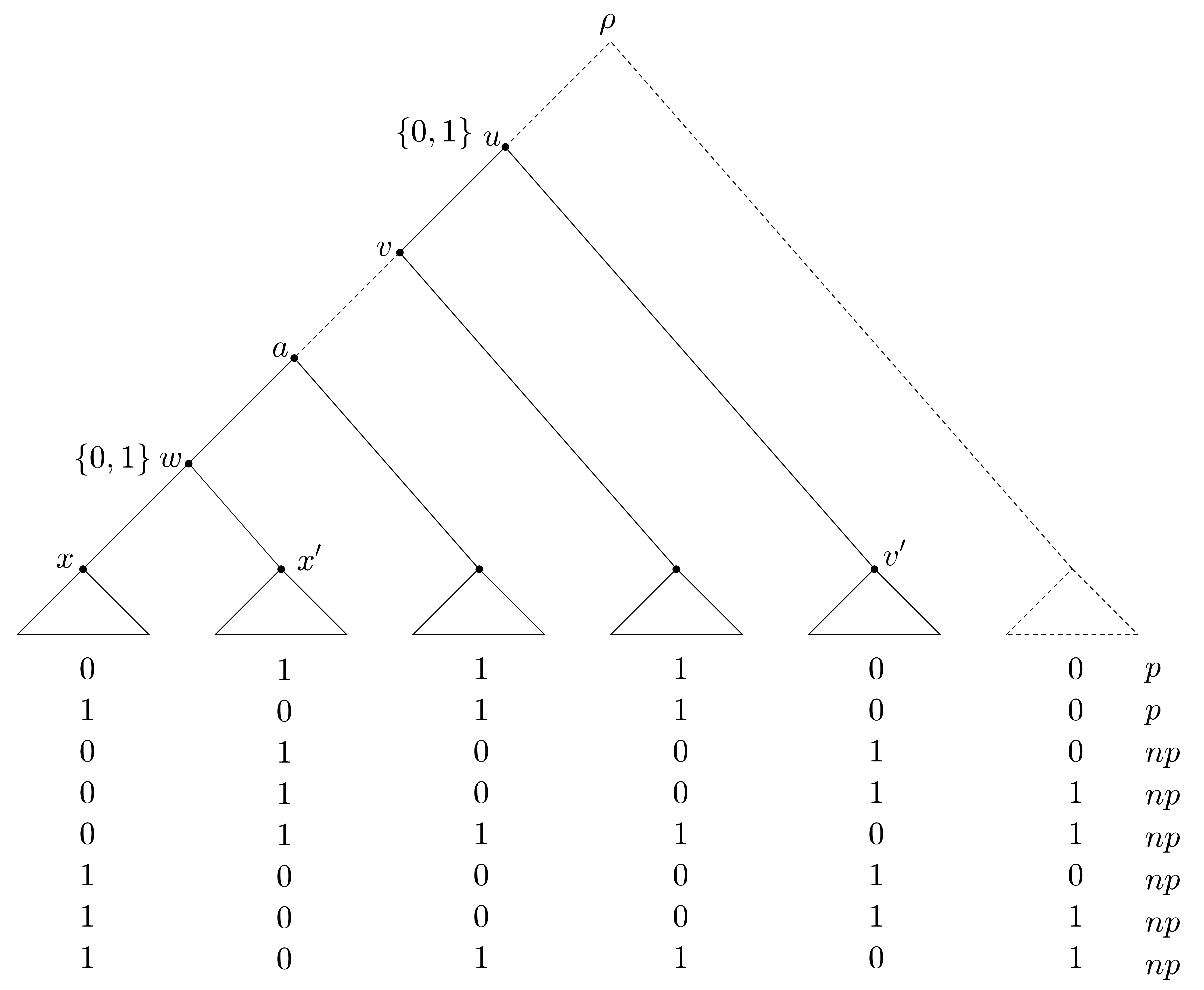}
	\caption{Characters used in the proof of Lemma \ref{01induces2}.}
	\label{Fig_01induces2}
\end{figure}

So in the light of Lemma \ref{01induces2}, in order to count persistent characters with parsimony score 2, we need to count the number of ways to pick $u$ and $w$ such that $u$ is an ancestor of $w$, but $w$ is not a child of $u$ (i.e. $u$ is not the direct ancestor of $w$). Every such choice will then immediately lead to two persistent characters with parsimony score 2, and there cannot be any more such characters.

Now as we want to count all pairs $\{u,w\}$ with the above properties, our main idea is to fix $u$ in order to count all possible choices of $w$ and to iterate this over all possible choices of $u$. However, before we can finally prove Theorem \ref{P2}, we need one more lemma, which already considers all choices of $u$ and considers the number of non-leaf children of $u$.

\begin{lemma} \label{d_u} Let $T$ be a rooted binary phylogenetic tree on $n \geq 2$ leaves. For each inner node $u$ of $T$, i.e. $u \in \mathring{V}(T)$, let $d_u$ denote the number of children of $u$ that are also inner nodes of $T$, i.e. $d_u$ can assume values 0, 1 or 2. Then, we have: $$\sum\limits_{u \in \mathring{V}(T)} d_u = n-2.$$
\end{lemma}

The proof of this lemma can be found in the appendix. Note that this lemma is surprising as it shows that  $\sum\limits_{u \in \mathring{V}(T)} d_u $ does not depend on the tree shape of $T$ but only on the number of leaves.

We now use Lemma \ref{d_u} to prove Theorem \ref{P2} and thus relate the Sackin index $\mathcal{S}(T)$ of a phylogenetic tree to its number of persistent characters, which is the main result of this section.

In the following, let $T$ be a rooted binary phylogenetic tree with $n$ leaves and root $\rho$, and let $v$ be an inner node of $T$. Then we denote by $T_v$ the subtree of $T$ rooted at $v$ and by $n_v$ the number of leaves of $T_v$. Note that then, $T_\rho = T$ and $n_\rho = n$. 

\begin{proof}[Proof of Theorem \ref{P2}]
\red{Let $T$ be a rooted binary phylogenetic tree with $n \geq 2$ leaves. We first consider the cases $n=2$ and $n=3$. For $n=2$, there is only one rooted binary tree shape that $T$ can have, and all $2^2=4$ binary characters are persistent on $T$. Furthermore, $\mathcal{S}(T)=2$, and thus we have $4=\mathcal{P}(T) = 2 \mathcal{S}(T)-2n+4=4-4+4$ as claimed. Similarly, for $n=3$, there is also only one rooted binary tree shape that $T$ can have, and all $2^3=8$ binary characters are persistent. Moreover, $\mathcal{S}(T)=5$, and we have $8 = \mathcal{P}(T) = 2 \mathcal{S}(T)-2n+4=10-6+4$ as claimed. Note that for $n=2, 3$, there are no binary characters with parsimony score 2. In particular, $\mathcal{P}_2(T)=0$ in this case.}

Now, let $n \geq 4$. We first consider $\mathcal{P}_2(T)$. By Lemma \ref{01induces2} we need to count all pairs $\{u,w\}$ such that $u$ is the ancestor of $w$ but $w$ is not a child of $u$. For a fixed $u$, this implies that we need to count all of its descendants that are inner nodes of $T$ except for its direct descendants. This can be done by considering $T_u$ and counting its number of inner nodes that are not adjacent to the root $u$ (because all inner nodes $w$ that are contained in $T_u$ are descendants of $u$, and if they are not adjacent to $u$ they are not direct descendants). 

Recall that every rooted binary phylogenetic tree $T_u$ with $n_u$ leaves has $n_u-1$ inner nodes including its root $\rho$, so it has $n_u-2$ inner nodes excluding its root. Moreover, it has $d_u$ inner nodes that are direct descendants of $u$. So the number of choices for $w$ for a given node $u$ is therefore $n_u-2-d_u$. 

This leads to \begin{equation} \label{eqP2}\mathcal{P}_2(T)= 2\left(\sum\limits_{u \in \mathring{V}(T)} (n_u-2-d_u)\right).\end{equation} (Note that here, the factor 2 is due to Lemma \ref{01induces2} \red{because every pair $\{u,w\}$ induces precisely two characters contributing to $\mathcal{P}_2(T)$}.) 

\noindent Moreover, we know that the number of summands is $n-1$, as $T$ has $n-1$ inner nodes (including $\rho$), so this leads to $\mathcal{P}_2(T)= 2\left(\sum\limits_{u \in \mathring{V}(T)} (n_u-2-d_u)\right) =  2\left(\sum\limits_{u \in \mathring{V}(T)} (n_u-d_u)-\sum\limits_{u \in \mathring{V}(T)} 2 \right) = 2\left(\sum\limits_{u \in \mathring{V}(T)} n_u-\sum\limits_{u \in \mathring{V}(T)}d_u-2(n-1) \right).$ \red{Now, by Lemma 6, $\sum\limits_{u \in \mathring{V}(T)} d_u=n-2$, and by definition of the Sackin index, $\sum\limits_{u \in \mathring{V}(T)} n_u = \mathcal{S}(T)$. Thus, in summary,} $\mathcal{P}_2(T)= 2 \left( \mathcal{S}(T) -(n-2)-2(n-1)\right) $. Expanding this term yields $\mathcal{P}_2(T)=2 \mathcal{S}(T)-6n+8$, which completes the proof for the formula for $\mathcal{P}_2(T)$.

Now for $\mathcal{P}(T)$, remember that by Lemma \ref{MPboundsPersistence}, Part 2, we know that $\mathcal{P}(T)=\mathcal{P}_0(T)+\mathcal{P}_1(T)+\mathcal{P}_2(T)$, and we have already seen that $\mathcal{P}_0(T)=2$ (given by the two constant characters) and $\mathcal{P}_1(T)=4n-6$ (given \red{by the characters corresponding to splits induced} by the inner edges of $T$ and the edges leading to the $n$ leaves of $T$). Thus, using the first part of the theorem, we derive $\mathcal{P}(T)=2+(4n-6)+(2\mathcal{S}(T)-6n+8).$ Expanding this term yields $\mathcal{P}(T) = 2 \mathcal{S}(T)-2n+4$, which completes the proof.
\end{proof}

Theorem \ref{P2} immediately leads to the following corollary.

\begin{corollary} Let $T_1$ and $T_2$ be two rooted binary phylogenetic trees with $n$ leaves. Then, we have:
\begin{enumerate}
\item $ \mathcal{S}(T_1) < \mathcal{S}(T_2) \Leftrightarrow  \mathcal{P}_2(T_1) < \mathcal{P}_2(T_2)\Leftrightarrow  \mathcal{P}(T_1) < \mathcal{P}(T_2).$
\item $ \mathcal{S}(T_1) = \mathcal{S}(T_2) \Leftrightarrow  \mathcal{P}_2(T_1) = \mathcal{P}_2(T_2)\Leftrightarrow  \mathcal{P}(T_1) = \mathcal{P}(T_2).$
\end{enumerate}
\end{corollary}

In other words, $T_1$ is more balanced than $T_2$ if and only if for $T_1$ there exist fewer persistent characters than for $T_2$. On the other hand, $T_1$ and $T_2$ are equally balanced if and only if they have the same number of persistent characters. We illustrate the first of these settings in the following example.

\begin{example}
As an example, we consider the caterpillar tree $T^{cat}_n$. 
The caterpillar tree maximizes the Sackin index and we have $\mathcal{S}(T_n^{cat})= \frac{n(n+1)}{2}-1$ (\citet{Fischer2018}). By Theorem \ref{P2} this leads to 
$$\mathcal{P}_2( T^{cat})= 2 S(T^{cat})-6n+8= 2\left(\frac{n(n+1)}{2}-1\right)-6n+8 = n^2-5n+6.$$
Now, suppose that $n$ is a power of $2$, i.e. $n=2^k$, for some ${k \geq 1},\, k \in \mathbb{N}$. 
Then we can compare the number of persistent characters with parsimony score 2 of the caterpillar tree on $n$ leaves with the number of these characters on the so-called fully balanced tree of height $k$. Recall that the Sackin index of the fully balanced tree is $S(T_k^{bal}) = k \cdot 2^k$ (\citet{Fischer2018}). By Theorem \ref{P2} this leads to $ \mathcal{P}_2(T^{bal}_k) = 2 \cdot (k\cdot 2^k)-6n+8.$ Using $n=2^k$, this leads to $$ \mathcal{P}_2(T^{bal}_k) = 2 \cdot(k\cdot 2^k)-6\cdot 2^k+8 = (k-3) \cdot 2^{k+1}+8.$$

When comparing the two examples $T_k^{bal}$ and $T^{cat}_n$ for $n=2^k$ and for $k\geq 1$, we observe the following: 

$$\mathcal{P}_2( T^{cat}_n)= n^2-5n+6= \left(2^{k}\right)^2-5\cdot 2^{k} +6 = 2^{2k} -5\cdot 2^k+6,$$ 
$$ \mathcal{P}_2(T^{bal}_k) = (k-3) \cdot 2^{k+1}+8. $$

It can easily be shown that for $k>1$, the first term is always strictly larger than the second. This implies that the number of persistent characters on the caterpillar tree, which equals  $\mathcal{P}_0(T^{cat}_n)+\mathcal{P}_1(T^{cat}_n)+\mathcal{P}_2(T^{cat}_n)$, is strictly larger than the number of persistent characters on the fully balanced tree with the same number of leaves, which equals $\mathcal{P}_0(T^{bal}_k)+\mathcal{P}_1(T^{bal}_k)+\mathcal{P}_2(T^{bal}_k)$. 
\end{example}

\begin{remark}
Note that as the caterpillar tree maximizes the Sackin index for all $n$ (and this maximum is unique; cf. \citet{Fischer2018}), there is no tree with more persistent characters. Moreover, for $n=2^k$, the Sackin index is minimized by the fully balanced tree of height $k$ (and again, this minimum is unique; cf. \citet{Fischer2018}), and thus, for $n=2^k$ there is no tree with fewer persistent characters than the fully balanced tree. For $n \neq 2^k$, there might be more than one tree minimizing the Sackin index, and thus, there might be more than one tree with a minimal number of persistent characters (the maximal number of persistent characters is always uniquely obtained on the caterpillar tree). We refer the reader to \citet{Fischer2018} for more details on the extremal values of the Sackin index if $n \neq 2^k$. 
However, for all $n \geq 2$ we can provide an upper and lower bound on the number of persistent characters, using the explicit bounds of the Sackin index stated in \citet{Fischer2018}.
\end{remark}

\begin{proposition} \label{PersBounds}
Let $T$ be a rooted binary phylogenetic tree with $n \geq 2$ leaves. Then we have
$$ -2^{\lceil \log_2 (n) \rceil+1} + 2n \lceil \log_2 (n) \rceil +4  \leq \mathcal{P}(T) \leq n^2-n+2.$$
\end{proposition}

\begin{proof}
By Theorem 1 and Theorem 3 in \citet{Fischer2018}, we have
$$ -2^{\lceil \log_2(n) \rceil} + n(\lceil \log_2(n) \rceil +1) \leq \mathcal{S}(T) \leq \frac{n(n+1)}{2}-1.$$
By Theorem \ref{P2} we know that $\mathcal{P}(T) = 2 \mathcal{S}(T) - 2n + 4$, and thus 
\begin{align*}
2 \big( -2^{\lceil \log_2(n) \rceil)} + n(\lceil \log_2(n) \rceil +1) \big) -2n + 4 &\leq \mathcal{P}(T) &\leq 2 \Big(\frac{n(n+1)}{2}-1 \Big) -2n + 4.
\end{align*}
Expanding all terms yields the desired result.
\end{proof}

In the previous two sections we have seen that persistent characters can be fully characterized by the first phase of the Fitch algorithm and that the number of characters that are persistent on a given tree $T$ depends on the shape of $T$. In the following section we will now consider the question of how characters together with their persistence status can be used to uniquely determine a tree. In particular, we consider the question of how many (carefully chosen) characters are needed to uniquely determine a tree. This question was posed as part of the \enquote{Kaikoura 2014 challenges} at the Kaikoura 2014 workshop (\citet{Kaikoura2014}), and we will provide an upper bound for this number.

\subsection{An upper bound on the minimum number of persistent characters that uniquely determine a tree} \label{sec_bound}

One of the earliest and most fundamental results in mathematical phylogenetics is the Buneman theorem (\citet{Bunemann1971}; see also \citet[p. 44]{Semple2003}), which basically states that an unrooted phylogenetic $X$-tree on leaf set $X$ with $|X|=n$ is uniquely defined by the set $\Sigma^*(T)$ of its induced non-trivial $X$-splits -- of which an unrooted binary tree has $n-3$. In particular, if such a set of compatible $X$-splits is given, the corresponding tree can be reconstructed in polynomial time using the so-called tree popping algorithm (\citet{Meacham1981, Meacham1983}).

Now recall that each $X$-split $\sigma$ can be translated into a binary character $f_\sigma$ (note that this character is unique if we assume that $f(1)=1$; otherwise we would also have to consider $\bar{f}_\sigma$). So if we translate the Buneman theorem into the setting of characters, it is obvious that an unrooted binary phylogenetic tree is uniquely determined by the set of $n-3$ binary characters that correspond to its $n-3$ inner edges. Considering parsimony, this can be interpreted in the following two (related) ways. Consider the alignment (set) $\mathcal{A}_T$ of the $n-3$ characters induced by some unknown unrooted binary phylogenetic \blue{$X$}-tree $T$.

\begin{itemize}
\item Assume you are given the information that for every $f \in \mathcal{A}_T$ we have $l(f,T)=1$, then you can reconstruct $\Sigma^*(T)$ and therefore also $T$ (via tree popping).
\item If you do not know $l(f,T)$ but are only given $\mathcal{A}_T$, theoretically you could consider all possible \blue{phylogenetic $X$-}trees $\widetilde{T}$ and calculate $l(\mathcal{A}_T,\widetilde{T})$. $T$ would then be the unique tree minimizing this score, i.e. $T= \argmin\limits_{\widetilde{T}} l(\mathcal{A}_T,\widetilde{T})$.
\end{itemize}

Note that the latter is due to the fact that all characters that employ two character states require at least one change, so it is clear that $l(f,\widetilde{T}) \geq 1$ for all $f\in \mathcal{A}_T$ and for all $\widetilde{T}$. Therefore, no other \blue{phylogenetic $X$-}tree can have a lower score than $T$. Moreover, as for different \blue{phylogenetic $X$-}trees $T$ and $T'$ we have $\Sigma^*(T)\neq \Sigma^*(T')$ and thus $\mathcal{A}_T \neq \mathcal{A}_{T'}$, but at the same time $|\Sigma^*(T)|=| \Sigma^*(T')|=n-3$, we know that at least one character $f \in \mathcal{A}_T $ is not contained in $\mathcal{A}_{T'}$ and thus has $l(f,T')\geq 2$. Thus, in total, for any $T' \neq T$, we have $l(\mathcal{A}_T,T')>l(\mathcal{A}_T,T)=n-3$. So $T$ is the unique Maximum Parsimony tree. 

In some sense, regarding the above two interpretations, the second one is stronger, because it leads to a reconstruction of $T$ based on $\mathcal{A}_T$ without any further information. For this procedure, i.e. reconstructing a tree according to the parsimony criterion based on a given alignment, there are many software tools available. However, note that finding a Maximum Parsimony tree is an NP-complete problem (\citet{Foulds1982}; see also \citet[p. 107]{Book_Steel}), and for large values of $n$, an exhaustive search through treespace is not applicable. (However, note that for our very specific alignment $\mathcal{A}_T$, which consists only of $n-3$ compatible binary characters, most software packages will still succeed in reconstructing $T$.)

The first interpretation, though, is in another sense more powerful, because it can make direct use of the additional information that $l(f,T)=1$ for all $f \in \mathcal{A}_T$ with the help of the tree popping algorithm. 

Now, coming back to persistence, it is natural to ask if similar characterizations exist for persistent characters: How many persistent characters do we need to uniquely determine a particular tree? 

Of course, this question might seem a bit powerless compared to parsimony at first if we consider only the second interpretation above. In particular, a character can only be persistent or not, whereas the parsimony score can assume various different values to indicate how good or bad the fit of the character to a given tree really is. In fact we will show subsequently that it is not sufficient to consider persistent characters in order to uniquely determine a rooted tree -- so compared to the second interpretation of the Buneman setting in terms of parsimony, persistent characters might seem weaker in reconstructing trees than most parsimonious ones. 

However, on the other hand, in the light of the first interpretation above, if we are allowed to list characters $f$ together with $\mathcalligra{p}(f,T)$, i.e. together with the information whether they are persistent or not, then we can indeed succeed in uniquely determining the tree -- and this is more powerful than parsimony in the sense that this even applies to rooted trees, whereas Maximum Parsimony can never distinguish between different root positions as it can only reconstruct unrooted trees. 
So in this sense, persistent characters are stronger than most parsimonious characters, but we will see that we need more of them to achieve this. 

However, we first consider the number of persistent characters needed to reconstruct unrooted trees, before we can turn our attention to rooted ones. Thus, in the following let $T^u$ denote the unrooted version of a rooted tree $T$, where $T^u$ is obtained from $T$ by suppressing the root node $\rho$ (i.e. by deleting $\rho$ and the two edges incident to it and re-connecting the two resulting degree-2 nodes with a new edge). 

We start this section with the first main theorem of this section.

\begin{theorem}[Buneman-type theorem for persistent characters] \label{bunemanpers}Let $T$ be a rooted binary phylogenetic $X$-tree with $|X|=n$. Let $T^u$ be its unrooted version with non-trivial split set $\Sigma^*(T^u)$. For each $\sigma \in \Sigma^*(T^u)$, let $f_\sigma$, $\bar{f}_\sigma$ denote the unique two binary characters induced by $\sigma$. Let $\mathcal{A}_T^p$ denote the alignment induced by all such characters $f_\sigma, \bar{f}_\sigma$ with $\sigma \in \Sigma^*(T^u)$. 

Then, if a rooted binary phylogenetic \blue{$X$-}tree $\widetilde{T}$ has the property that all $ f \in \mathcal{A}_T^p$ are persistent on $\widetilde{T}$, we have $\widetilde{T}^u=T^u$. 

In other words, the unrooted version $T^u$ of $T$ is uniquely determined by $\mathcal{A}_T^p$.

\end{theorem}

In order to prove this theorem, we need two more lemmas. The first one shows that if $f$ is persistent on a tree $T$ and has parsimony score 2, then its inverted counterpart $\bar{f}$ cannot be persistent.

\begin{lemma} \label{fbarnonpers} Let $T$ be a rooted binary phylogenetic tree, let $f$ be a binary character with $l(f,T)=2$. Then, we have:    $$\mathcalligra{p}(f,T)=p \Rightarrow \mathcalligra{p}(\bar{f},T)=np.$$
\end{lemma}
%%%
\begin{proof} As parsimony does not distinguish between $f$ and $\bar{f}$, we have $l(\bar{f},T)=l(f,T)=2$. Thus, both characters will have two $\{0,1\}$ union sets assigned to some nodes $u$ and $w$ during the first phase of the Fitch algorithm, and these nodes are identical for $f$ and $\bar{f}$. As $f$ is persistent, $f$ is by Lemma \ref{01induces2} one of only two characters that is persistent on $T$ and employs these particular nodes $u$ and $w$ for the $\{0,1\}$ sets. Moreover, by Remark \ref{atleastoneinbetween}, there is at least one node between $u$ and $w$, and as in the proof of Theorem \ref{characterize}, Part 3, we can conclude that $f$ has a unique minimal persistent extension that assigns this node state $1$. As there cannot be any other $\{0,1\}$ set in $T$, this implies that all leaves descending from this node must be in state $1$. Note that the same would have to apply to $\bar{f}$ if $\bar{f}$ was persistent. But this cannot be the case as $\bar{f}$ assigns 0 to precisely those leaves to which $f$ assigns 1. So $\bar{f}$ cannot be persistent. This completes the proof.
\end{proof}

Next, we consider again $f$ and $\bar{f}$, but for the case $l(f,T)=l(\bar{f},T)=1$.

\begin{lemma}\label{score1characterized}
Let $T$ be a rooted binary phylogenetic $X$-tree and let $f$ be a non-constant binary character on $X$. Then, $$l(f,T)=1 \Longleftrightarrow \mathcalligra{p}(f,T)=\mathcalligra{p}(\bar{f},T)=p.$$
\end{lemma}

\begin{proof}
Let $l(f,T)=1$. Then, $l(\bar{f},T)=1$ and thus, by Theorem \ref{characterize}, Part 2, both $f$ and $\bar{f}$ are persistent on $T$, i.e. we have $ \mathcalligra{p}(f,T)=\mathcalligra{p}(\bar{f},T)=p. $

Conversely, if $ \mathcalligra{p}(f,T)=\mathcalligra{p}(\bar{f},T)=p$, i.e. if both $f$ and $\bar{f}$ are persistent on $T$, then by the second part of Lemma \ref{MPboundsPersistence}, we have $l(f,T)=l(\bar{f},T)\leq 2$. As $f$ (and thus $\bar{f}$) is not constant by assumption, we also have $l(f,T)=l(\bar{f},T)\geq 1$. Now if we had $l(f,T)=l(\bar{f},T)=2$, by Lemma \ref{fbarnonpers}, one of the two characters $f$, $\bar{f}$ could not be persistent. But as both are persistent by assumption, we conclude $l(f,T)=l(\bar{f},T)= 1$, which completes the proof.
\end{proof}

Thus, we now know that all characters which have parsimony score 1 (and thus correspond to splits in the Bunemann setting) are exactly those characters $f$, where both $f$ and its inverted version $\bar{f}$ are persistent on $T$. This allows us to prove Theorem \ref{bunemanpers}.

\begin{proof} [Proof of Theorem \ref{bunemanpers}] First note that $\mathcal{A}_T^p$ does not contain any constant characters as it is based on the splits of $\Sigma^*(T^u)$. So for each $f \in \mathcal{A}_T^p$, we have $l(f,\hat{T})\geq 1$ for each binary phylogenetic $X$-tree $\hat{T}$. Moreover, as we consider only splits from $\Sigma^*(T^u)$, i.e. non-trivial splits, we know that all $f \in \mathcal{A}_T^p$ are informative, \red{i.e. they contain at least two zeros and two ones.}

Now let $\widetilde{T}$ be such that all $ f \in \mathcal{A}_T^p$ are persistent on $\widetilde{T}$. Note that by construction,  for each character $f \in \mathcal{A}_T^p$,  $\mathcal{A}_T^p$ contains also its inverted character $\bar{f}$. Now, as for each such pair $f$, $\bar{f}$ we have persistence, by Lemma \ref{score1characterized}, we conclude that $l(f,\widetilde{T})=l(\bar{f},\widetilde{T})=1$. Again, this implies that there is a $\sigma \in \Sigma^*(\widetilde{T}^u)$ such that $f$ and $\bar{f}$ correspond to $\sigma$ (note that $\sigma$ cannot be in $\Sigma(\widetilde{T}^u)\setminus \Sigma^*(\widetilde{T}^u)$, because $f$ is informative). As this by assumption holds for all $ f \in \mathcal{A}_T^p$, and as $|\Sigma^*(\widetilde{T}^u)|=|\Sigma^*(T^u)|=n-3$, we conclude that $\Sigma^*(\widetilde{T}^u)=\Sigma^*(T^u)$ which, by the Buneman theorem (\citet{Bunemann1971}), implies that $\widetilde{T}^u = T^u$. This completes the proof.
\end{proof}

Theorem \ref{bunemanpers} immediately leads to the following corollary.

\begin{corollary} \label{2n-3suffice}Let $T$ be a rooted binary phylogenetic tree with $n$ leaves. Then, listing all $2(n-3)$ persistent characters $f$ of $T$ that correspond to $\Sigma^*(T)$ together with their persistence status $\mathcalligra{p}(f,T)=p$ suffices to uniquely determine $T^u$. 
\end{corollary}

\begin{proof} The statement is a direct consequence of Theorem \ref{bunemanpers}. 
\end{proof}

So while we now know that $2(n-3)$ characters together with their persistence status suffice to fix the unrooted version of $T$, it can easily be seen that these characters do not suffice to fix $T$. We illustrate this with a simple example. 

\begin{example}
Consider again the two \blue{phylogenetic $X$-}trees $T_1$ and $T_2$ with $n=4$ leaves as depicted in Figure \ref{Fig_TreesUnrooted}. Note that $T_1^u = T_2^u=T$ as depicted in Figure \ref{Fig_TreesUnrooted}. 

\begin{figure}[htbp]
	\centering
	\includegraphics[scale=0.1]{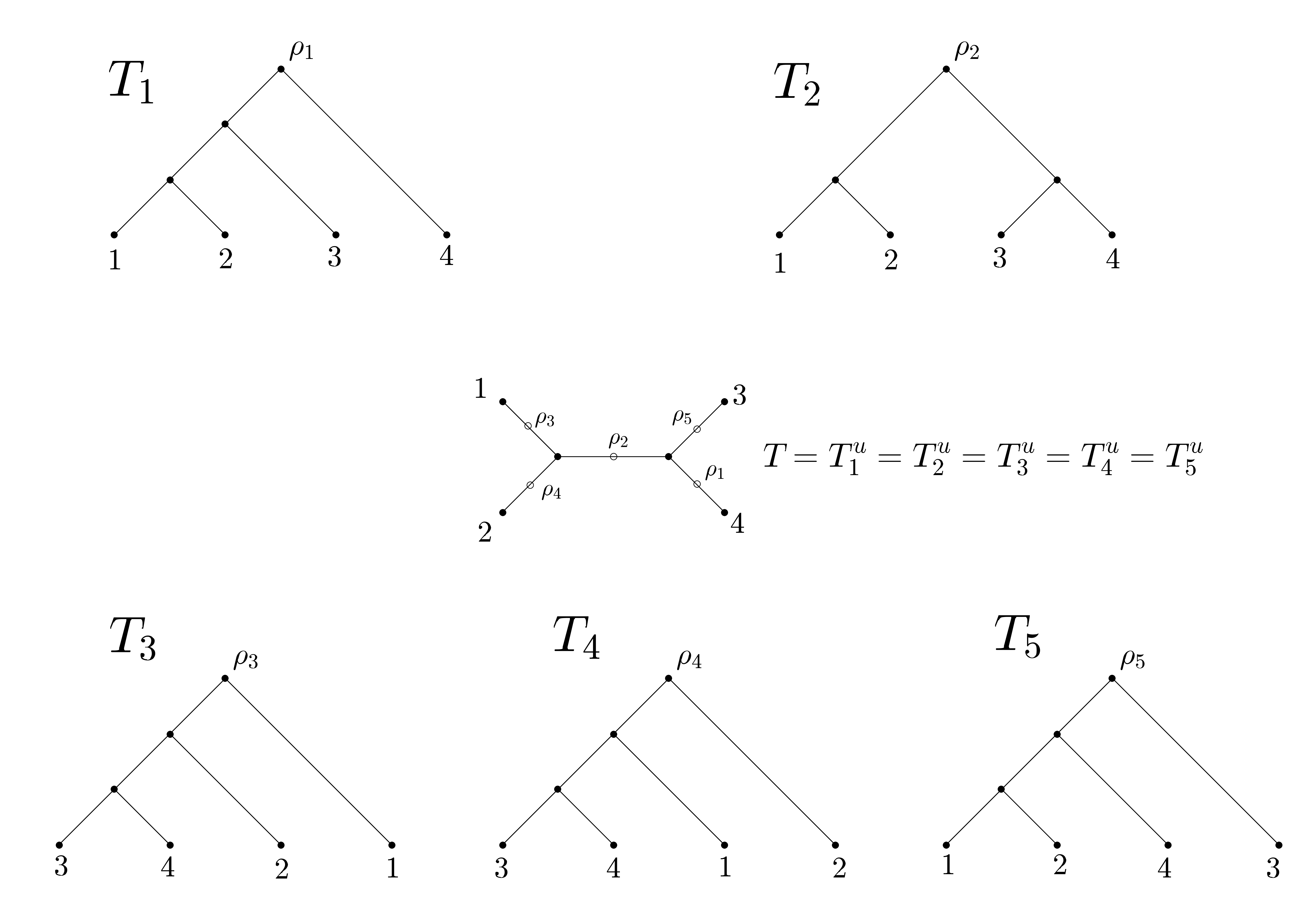}
	\caption{Tree $T$: Unrooted version of trees $T_1$ and $T_2$ as depicted in Figure \ref{Fig_TreesP2} as well as of trees $T_3, \, T_4$ and $T_5$. Note that on all rootings of $T$ the characters $f=1100$ and $\bar{f}=0011$ are persistent.}
	\label{Fig_TreesUnrooted}
\end{figure}

Thus, we have $\Sigma^*(T_1^u) = \Sigma^*(T_2^u) =\Sigma^*(T) =\{ 12|34\}$. Therefore, the $2(n-3)=2(4-3)=2$ characters that correspond to $\Sigma^*(T_1^u)$ and $ \Sigma^*(T_2^u) $ are $f=1100$ and $\bar{f}=0011$. Note that $f$ and $\bar{f}$ are persistent on $T_1$ and $T_2$. So if we were given these two characters together with their persistence status, i.e. with the information that they are persistent on the tree that we seek, we still could not distinguish between $T_1$ and $T_2$ (in fact, there are three more \blue{phylogenetic $X$-}trees -- namely the other three of the five possible rootings of $T$ -- where both of these characters are persistent, see Figure \ref{Fig_TreesUnrooted}). But there is no rooted binary phylogenetic \blue{$X$-}tree $\widehat{T}$ with $ \widehat{T}^u \neq T^u$ on which both $f$ and $\bar{f}$ are persistent. As an example, consider tree $\widehat{T}$ as depicted in Figure \ref{Fig_That}, whose unrooted version $ \widehat{T}^u$ is also depicted in Figure \ref{Fig_That} and does not equal $T^u$. On $\widehat{T}$, only $f$ is persistent (the unique persistent extension is depicted in Figure \ref{Fig_That}), but $\bar{f}$ is not. This is due to the fact that $l(f,\widehat{T})=l(\bar{f},\widehat{T})=2$ and thus, according to Lemma \ref{fbarnonpers}, if $f$ is persistent, $\bar{f}$ cannot be persistent.

\begin{figure}[htbp]
	\centering
	\includegraphics[scale=0.1]{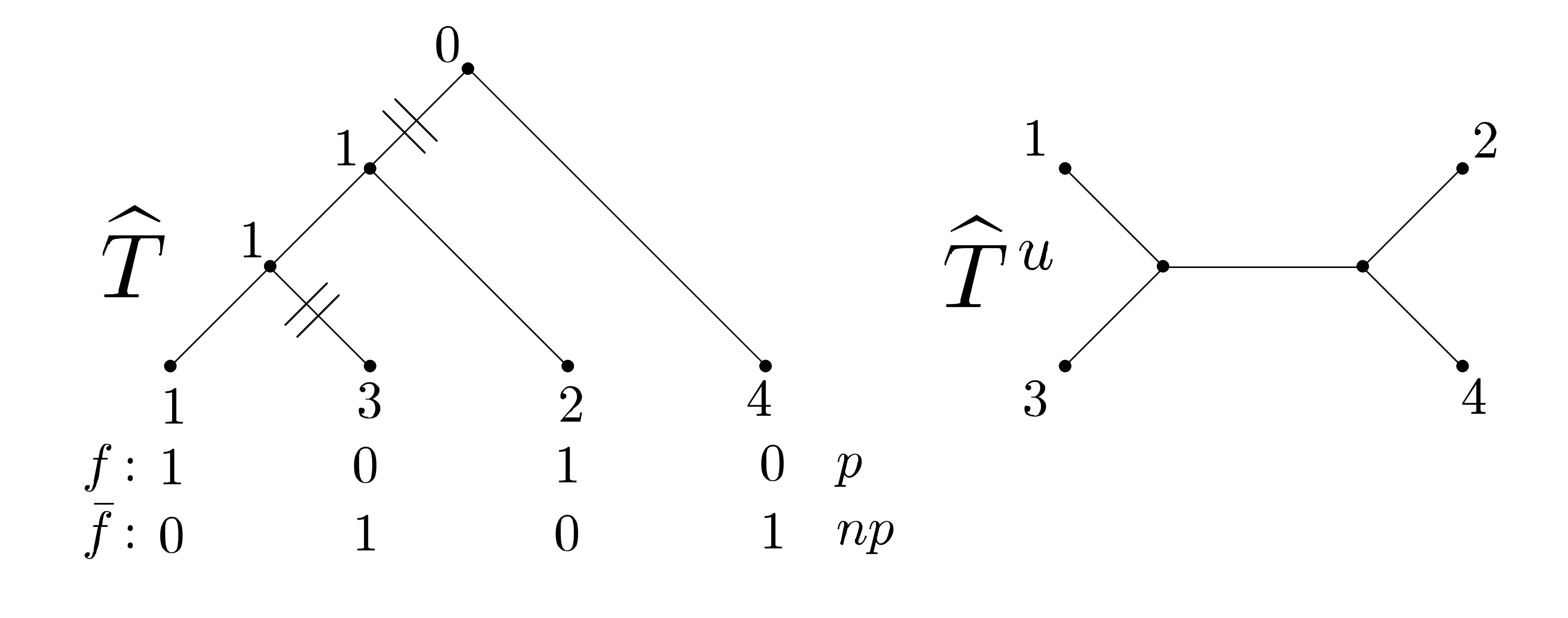}
	\caption{Tree $\widehat{T}$ and its unrooted version $\widehat{T}^u$, on which only $f=1100$ is persistent, while $\bar{f}=0011$ is not. A persistent extension for $f$ is obtained by assigning state $0$ to the root and state $1$ to all other internal nodes of $\widehat{T}$.}
	\label{Fig_That}
\end{figure}

\end{example}

Moreover, this example shows that not even {\em all} persistent characters of a rooted binary phylogenetic tree $T$ together with their persistence status necessarily suffice to uniquely determine $T$. Recall that by Section \ref{sec_shape} we know that $\mathcal{P}_0(T_1)=\mathcal{P}_0(T_2)=2$ and $\mathcal{P}_1(T_1)=\mathcal{P}_1(T_2)=4n-6 = 10$. Moreover, by Theorem \ref{P2} we know that $\mathcal{P}_2(T_2)= 2 S(T_2)-6n+8 = 2(2+2+4)-6 \cdot 4 +8 =0 $, i.e. $T_2$ has no persistent character of parsimony score 2. So {\em all} characters that are persistent on $T_2$ are also persistent on $T_1$, which is why listing all persistent characters of $T_2$ cannot suffice to uniquely determine $T_2$. (Note, however, that $T_1$, on the other hand \emph{is} uniquely determined by its $2+10+2=14$ persistent characters, which can be easily checked).

So in order to fix not only the unrooted but even the rooted version of a binary phylogenetic tree $T$, we need some information provided by the non-persistent characters, too. The main aim of the remainder of this note will therefore be to show that $2n-3$ (carefully chosen) characters suffice to uniquely determine a rooted phylogenetic tree. This is summarized by the following theorem.

\begin{theorem}\label{bound} Let $T$ be a rooted binary phylogenetic \blue{$X$-}tree with $n$ leaves. Then, it is possible to list $2(n-3)+3 = 2n-3$ characters together with their respective persistence status in order to distinguish $T$ from any other rooted binary phylogenetic \blue{$X$-}tree $\widetilde{T}$; i.e. the characters of this list will not all have the same persistence status as on $T$ on any other \blue{phylogenetic $X$-}tree $\widetilde{T}$. Thus, they uniquely determine $T$.
\end{theorem}

In this context, we first consider the following lemma, which states that three characters together with their persistence status suffice to uniquely determine the root position for a given unrooted phylogenetic tree.

\begin{lemma} \label{4suffice} Let $T$ be a rooted binary phylogenetic \blue{$X$-}tree with root $\rho$ and unrooted version $T^u$. 
\begin{enumerate}
\item If $T$ has four leaves, i.e. $n=4$, then two (carefully chosen) characters together with their persistence status suffice in order to distinguish $T$ from any other \blue{phylogenetic $X$-}tree $\widetilde{T}$ with $\widetilde{T}^u=T^u$.
\item If the number $n$ of leaves of $T$ is at least 5, i.e. $n\geq 5$, then three (carefully chosen) characters together with their persistence status suffice in order to distinguish $T$ from any other \blue{phylogenetic $X$-}tree $\widetilde{T}$ with $\widetilde{T}^u=T^u$.
\end{enumerate}
\end{lemma}

\begin{proof} \leavevmode
\begin{enumerate}
\item We first consider the case of four leaves. 
	Recall that for $n=4$, there are two different tree shapes, namely the caterpillar tree $T_{4}^{cat}$ and the  fully balanced tree $T_2^{bal}$ of height 2. We will consider both of them separately and show that in any case two characters together with their persistence status suffice to correctly determine the root position from the unrooted versions of $T_{4}^{cat}$ and $T_2^{bal}$. 

Consider $T_4^{cat}$ as depicted in Figure \ref{Fig_4Taxa_1}.
Then, e.g. $f_1 = 1001$ and $f_2 = 1010$ together with their persistence status suffice to distinguish $T_4^{cat}$ from any of the four other possible rootings of $T_4^{{cat}^u}$. We have $\mathcalligra{p}(f_1, T_4^{cat})=np$ and $\mathcalligra{p}(f_2, T_4^{cat})=p$ and there is no other rooting of $T_4^{{cat}^u}$ such that $f_2$ is persistent, while $f_1$ is not. 

\begin{figure}[h!]
	\centering
	\includegraphics[scale=0.15]{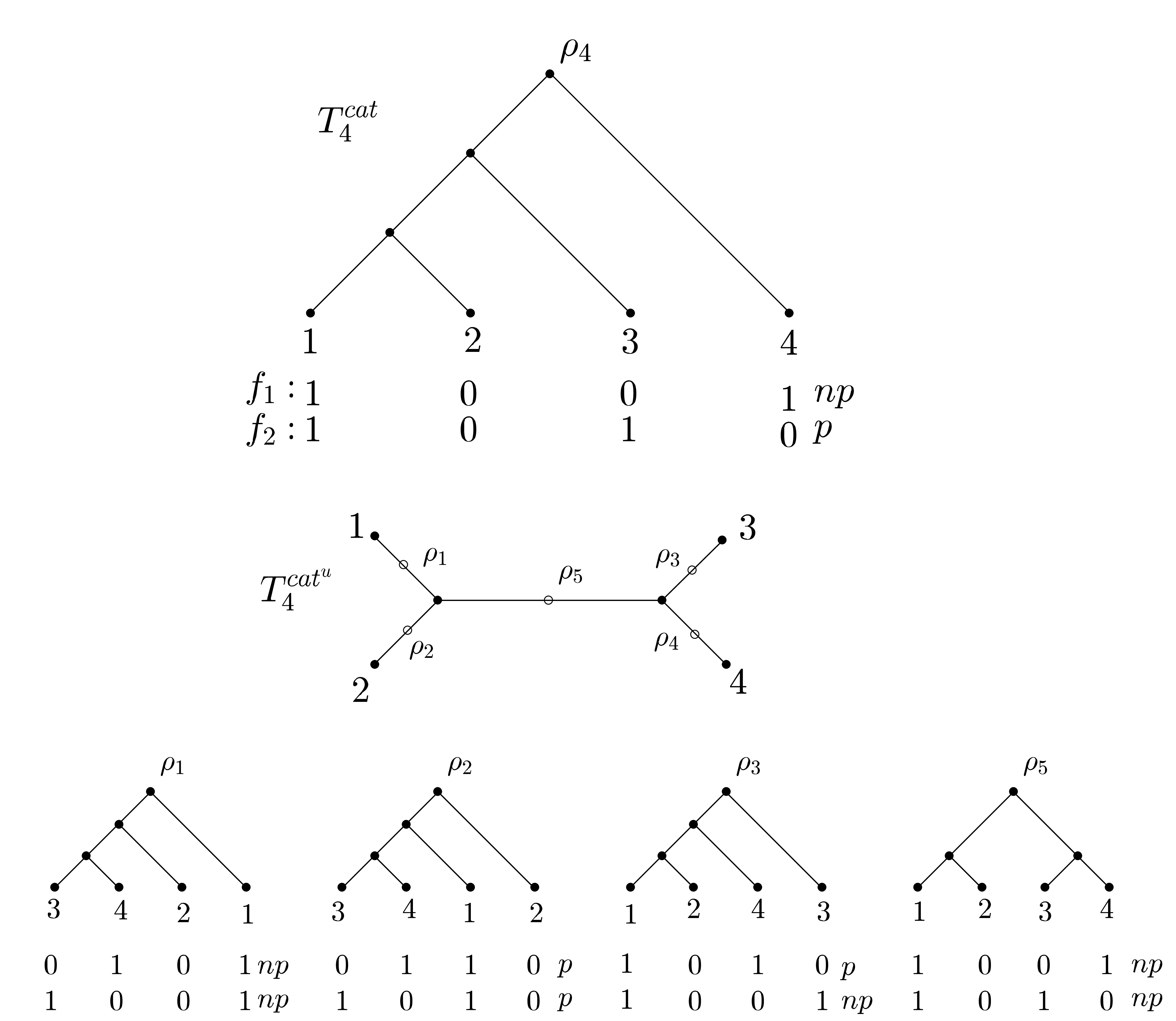}
	\caption{Caterpillar tree $T_4^{cat}$ and its unrooted version. The two characters $f_1$ and $f_2$ together with their persistence status suffice to distinguish $T_4^{cat}$ from any of the other rootings. }
	\label{Fig_4Taxa_1}
\end{figure}

Now, consider $T_2^{bal}$ as depicted in Figure \ref{Fig_4Taxa_2}.
We set $f_1=0101$ and $f_2 = 1010$. Then, we have $\mathcalligra{p}(f_1, T_2^{bal}) = \mathcalligra{p}(f_2, T_2^{bal})=np$ and there is no other rooting of $T_2^{{bal}^u}$ such that $f_1$ and $f_2$ are both not persistent. 

\begin{figure}[h!]
	\centering
	\includegraphics[scale=0.15]{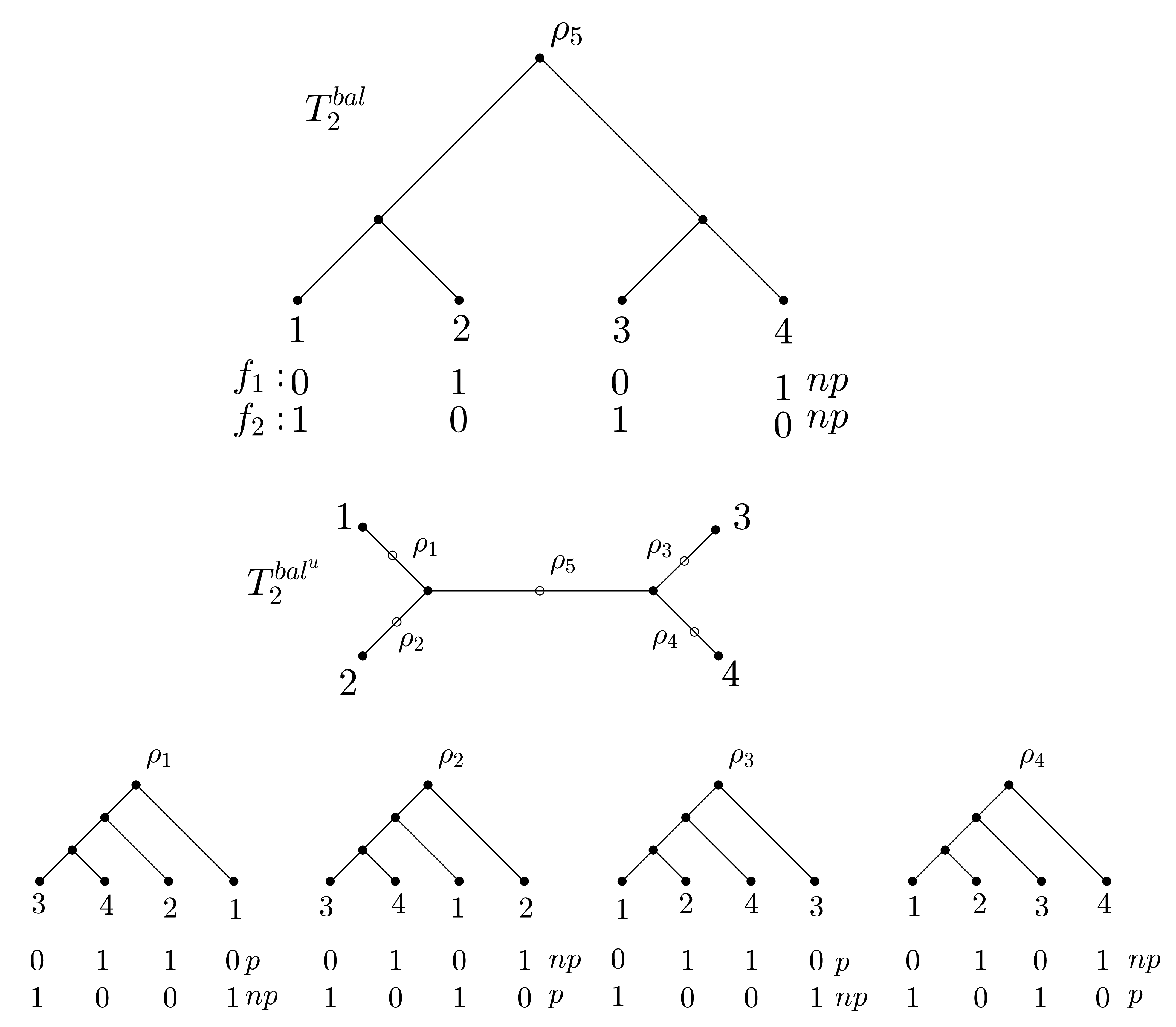}
	\caption{Fully balanced tree $T_2^{bal}$ and its unrooted version. The two characters $f_1$ and $f_2$ together with their persistence status suffice to distinguish $T_2^{bal}$ from any of the other rootings.}
	\label{Fig_4Taxa_2}
\end{figure}

Thus, in both cases two characters together with their persistence status suffice in order to distinguish a tree $T$ on four leaves from any other tree $\widetilde{T}$ on four leaves with $\widetilde{T}^u = T^u$.

\item In the case of $n \geq 5$, we provide an explicit construction of the three characters and their respective persistence status. 

It can be easily seen that every binary rooted tree with $n \geq 5$ leaves has height $h(T) \geq 3$. Note that the root $\rho$ has two direct descendants, say $a$ and $b$. We denote the maximal pending subtrees rooted at these nodes as $T_a$ and $T_b$ and their number of leaves as $n_a$ and $n_b$, respectively. 

Consider the longest path from $\rho$ to any cherry in $T$. We refer to the leaves of this cherry as $x_1$ and $x_2$, respectively. Then, we assume without loss of generality that both $x_1$ and $x_2$ are contained in $T_a$ (otherwise reverse the roles of $T_a$ and $T_b$). Due to $h(T) \geq 3$, $T_a$ has at least one more taxon other than $x_1$ and $x_2$. $T_a$ can thus be further subdivided into its two maximal pending subtrees $T_a^1$ and $T_a^2$. Without loss of generality, we assume that $x_1$ and $x_2$ are part of $T_a^1$.

Now we now construct $f_1$ with $\mathcalligra{p}(f_1,T)=p$ as follows:

\begin{itemize}
\item We set $f_1(x_1)=0$.
\item We set $f_1(x_2)=1$.
\item We set $f_1(x)=1$ for all $x$ that are leaves of $T_a$ but $x \neq x_1,x_2$.
\item We set $f_1(x)=0$ for all $x$ that are leaves of $T_b$.
\end{itemize} 

Note that by construction, we have $l(f_1,T)=2$ and $\mathcalligra{p}(f_1,T)=p$ (Figure \ref{Fig_4suffice_1}).

\begin{figure}[htbp]
	\centering
	\includegraphics[scale=0.1]{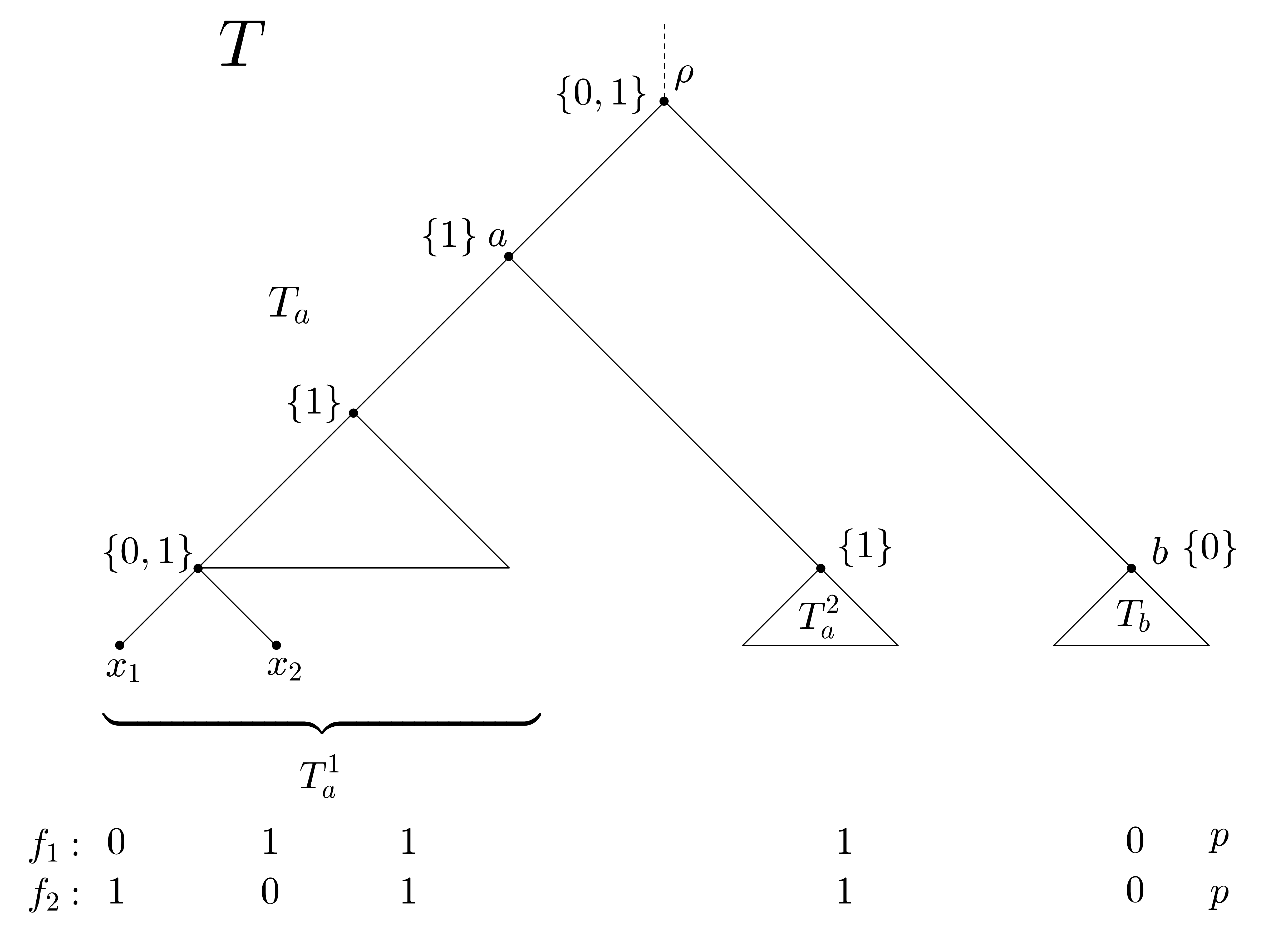}
	\caption{Phylogenetic \blue{$X$-}tree $T$ consisting of subtrees $T_a$ and $T_b$ together with characters $f_1$ and $f_2$. Note that $f_1$ and $f_2$ are both persistent on $T$. For $f_1$ we require a $0 \rightarrow 1$ change on the edge leading from $\rho$ to $a$ and a $1 \rightarrow 0$ change on the edge leading to $x_1$. Analogously, for $f_2$ we require a $ 0 \rightarrow 1$ change on the edge leading to $a$ and a $1 \rightarrow 0$ change on the edge leading to $x_2$. Furthermore, $f_1$ and $f_2$ have parsimony score 2, because the first phase of the Fitch algorithm would assign the state set $\{0,1\}$ to both the parent of $x_1$ and $x_2$ and the root $\rho$ and thus, we have $l(f_1,T)=l(f_2,T)=2$.}
	\label{Fig_4suffice_1}
\end{figure}

Given $T^u$ as depicted in Figure \ref{Fig_4suffice_2} and a persistent character such as $f_1$ with $l(f_1,T^u)=2$, we know by Theorem \ref{P2} that the two $\{0,1\}$ union sets that the Fitch phase will assign to inner nodes during the first phase (as $l(f_1,T^u)=2$) have to be on a common path from $\rho$ to one of the leaves, but they cannot be directly adjacent. Therefore, considering Figure \ref{Fig_4suffice_2}, $f_1$ already gives us some hints concerning the position of $\rho$: As the $0 \rightarrow 1$ change has to happen before the $1 \rightarrow 0$ change, $\rho$ could be placed on any edge of $T_b$ or on the edge $e=\{a,b\}$ connecting $T_a$ and $T_b$ in $T^u$ or on the edge on which $x_1$ is pending. It could not, however, be placed on any other edge of $T_a$, because otherwise $f_1$ would not be persistent.

\begin{figure}[htbp]
	\centering
	\includegraphics[scale=0.1]{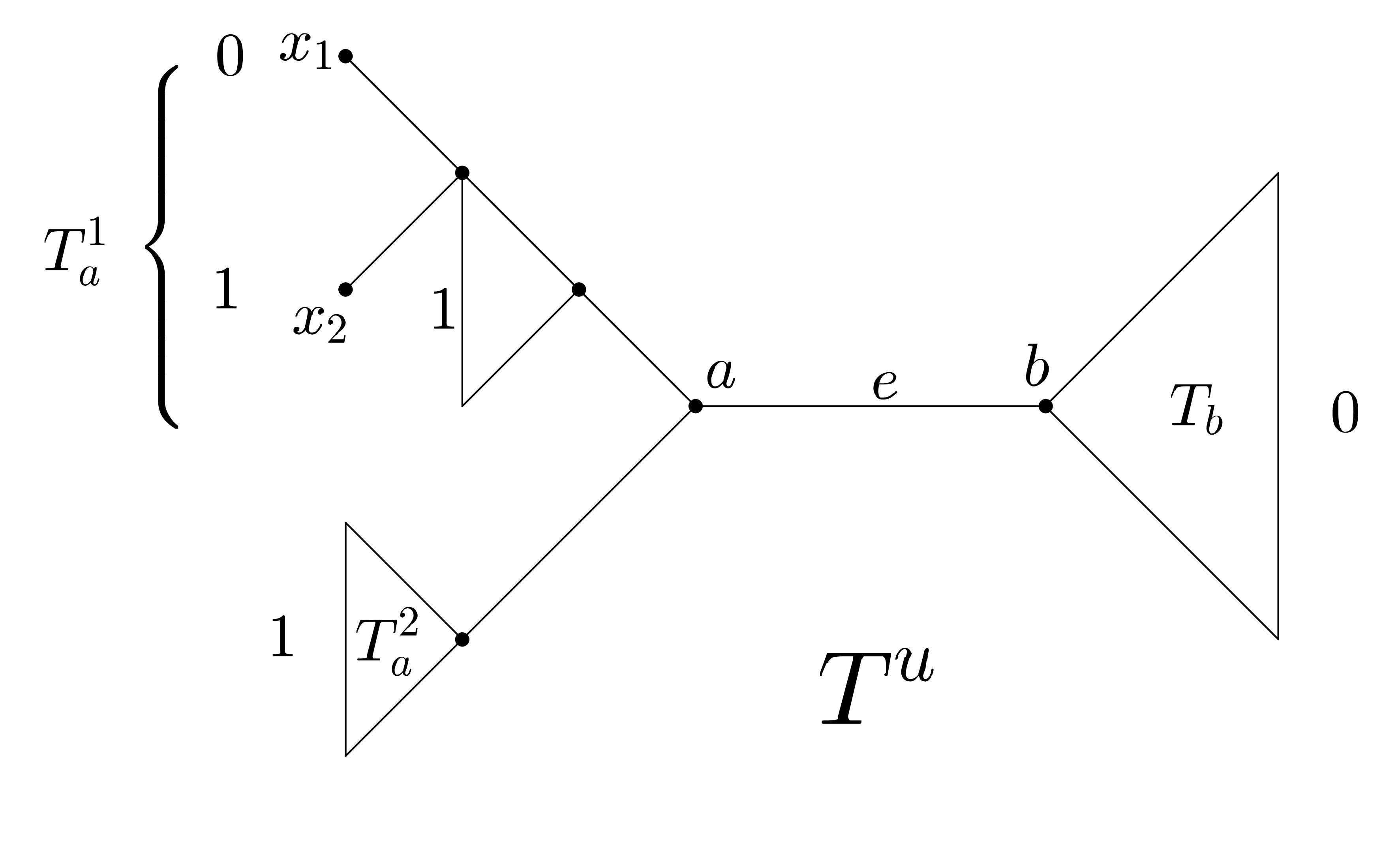}
	\caption{Unrooted version of the \blue{phylogenetic $X$-}tree $T$ depicted in Figure \ref{Fig_4suffice_1} together with character $f_1$.}
	\label{Fig_4suffice_2}
\end{figure}

Now in order to exclude the possibility that the root is wrongly positioned on the edge leading to $x_1$, we construct character $f_2$ with $\mathcalligra{p}(f_2,T)=p$ as follows:

\begin{itemize}
\item We set $f_2(x_1)=1$.
\item We set $f_2(x_2)=0$.
\item We set $f_2(x)=1$ for all $x$ that are leaves of $T_a$ but $x \neq x_1,x_2$.
\item We set $f_2(x)=0$ for all $x$ that are leaves of $T_b$.
\end{itemize} 

Note that by construction, we have $l(f_2,T)=2$ and $\mathcalligra{p}(f_2,T)=p$, i.e. $f_2$ is  persistent on $T$ (Figure \ref{Fig_4suffice_1}).

This indeed excludes the root position on the edge leading to $x_1$, because on that rooted version of $T^u$, $f_2$ would {\em not} be persistent as can be seen in Figure \ref{Fig_4suffice_3}.

\begin{figure}
	\centering
	\includegraphics[scale=0.1]{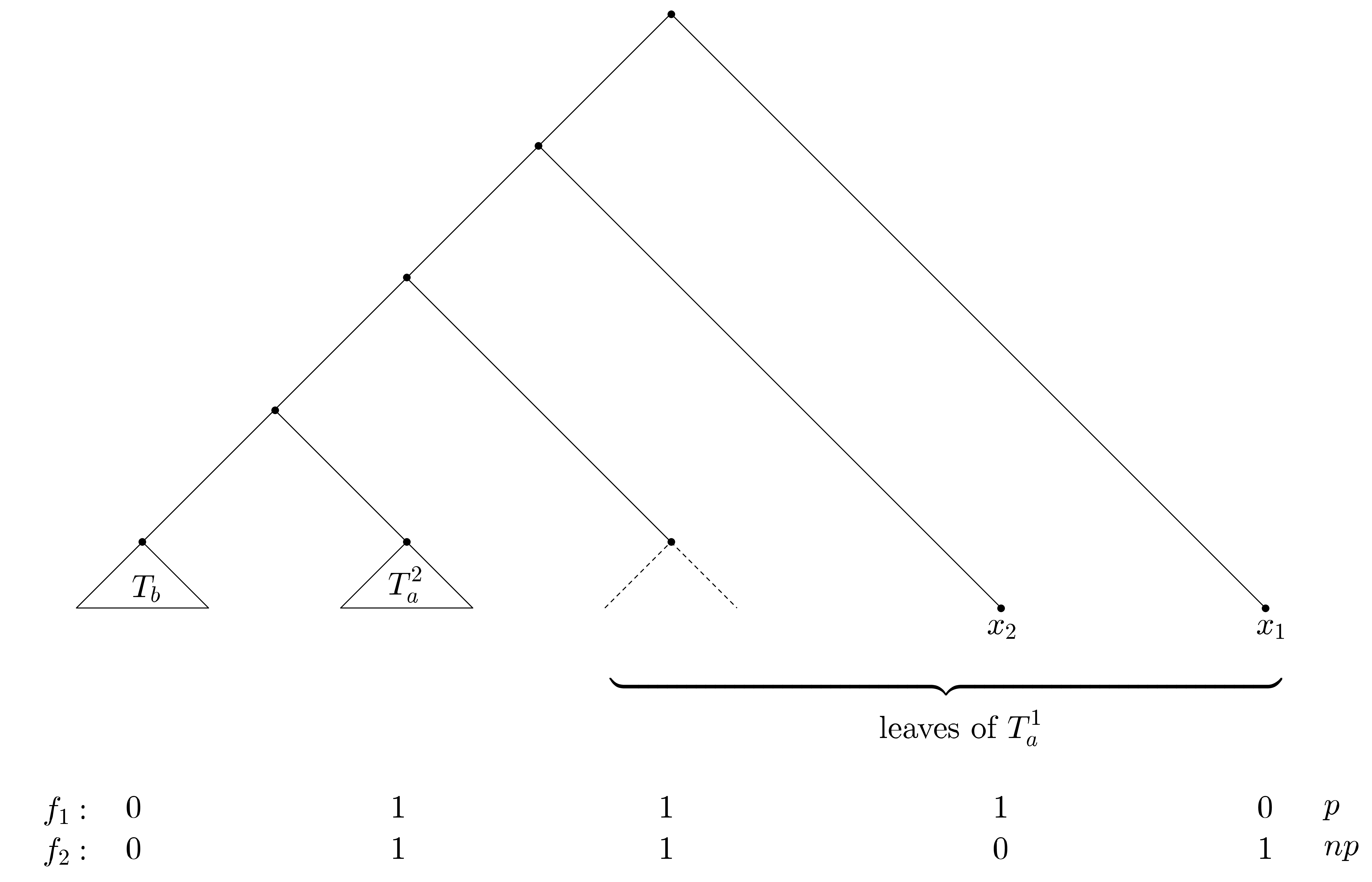}
	\caption{\blue{Phylogenetic $X$-}tree $T^u$ as depicted in Figure \ref{Fig_4suffice_2} rooted on the edge leading to leaf $x_1$ together with characters $f_1$ and $f_2$.}
	\label{Fig_4suffice_3}
\end{figure}

Note that the only possible root positions in $T^u$ are now $e=\{a,b\}$ \red{(i.e. the correct root position of $T$)} and all edges in $T_b$. So if $T_b$ consists of only one node (i.e. $T_b$ contains no edge), we are are already done -- in this case, two characters, namely $f_1$ and $f_2$ suffice to fix the root position (this is for instance the case when $T$ is a caterpillar), because the only remaining root position would be the correct edge $e=\{a,b\}$.

Now if $T_b$ has at least two leaves, we will see that two more characters suffice to exclude the edges in $T_b$ as possible root positions. We will also see that these  two automatically exclude the possibility that the root is wrongly positioned on the edge leading to $x_1$ and make character $f_2$ redundant in this case. In order to show that we note that as $T_b$ has at least two leaves, $T_b$ can also be subdivided into its two maximal pending subtrees $T_b^1$ and $T_b^2$.

Now we construct the following characters $f_3$ and $f_4$ with $\mathcalligra{p}(f_3,T)=\mathcalligra{p}(f_4,T)=np$:

\begin{itemize}
\item We set $f_3(x)=f_4(x)=0$ for all leaves $x$ of $T_a^1$.
\item We set $f_3(x)=f_4(x)=1$ for all leaves $x$ of $T_a^2$.
\item We set $f_3(x)=0$ for all leaves $x$ of $T_b^1$.
\item We set $f_3(x)=1$ for all leaves $x$ of $T_b^2$.
\item We set $f_4(x)=1$ for all leaves $x$ of $T_b^1$.
\item We set $f_4(x)=0$ for all leaves $x$ of $T_b^2$.
\end{itemize} 

Note that by construction, we have $l(f_3,T)=l(f_4,T)=2$ \red{(because applying the first phase of the Fitch algorithm for $f_3$, respectively $f_4$, will result in two $\{0,1\}$ union nodes, namely $a$ and $b$. Note that in this case the root of $T$ will be a $\{0,1\}$ intersection node)} and $\mathcalligra{p}(f_3,T)=\mathcalligra{p}(f_4,T)=np$, i.e. $f_3$ and $f_4$ are {\em not} persistent on $T$ (Figure \ref{Fig_4suffice_6}).

\begin{figure}[htbp]
	\centering
	\includegraphics[scale=0.1]{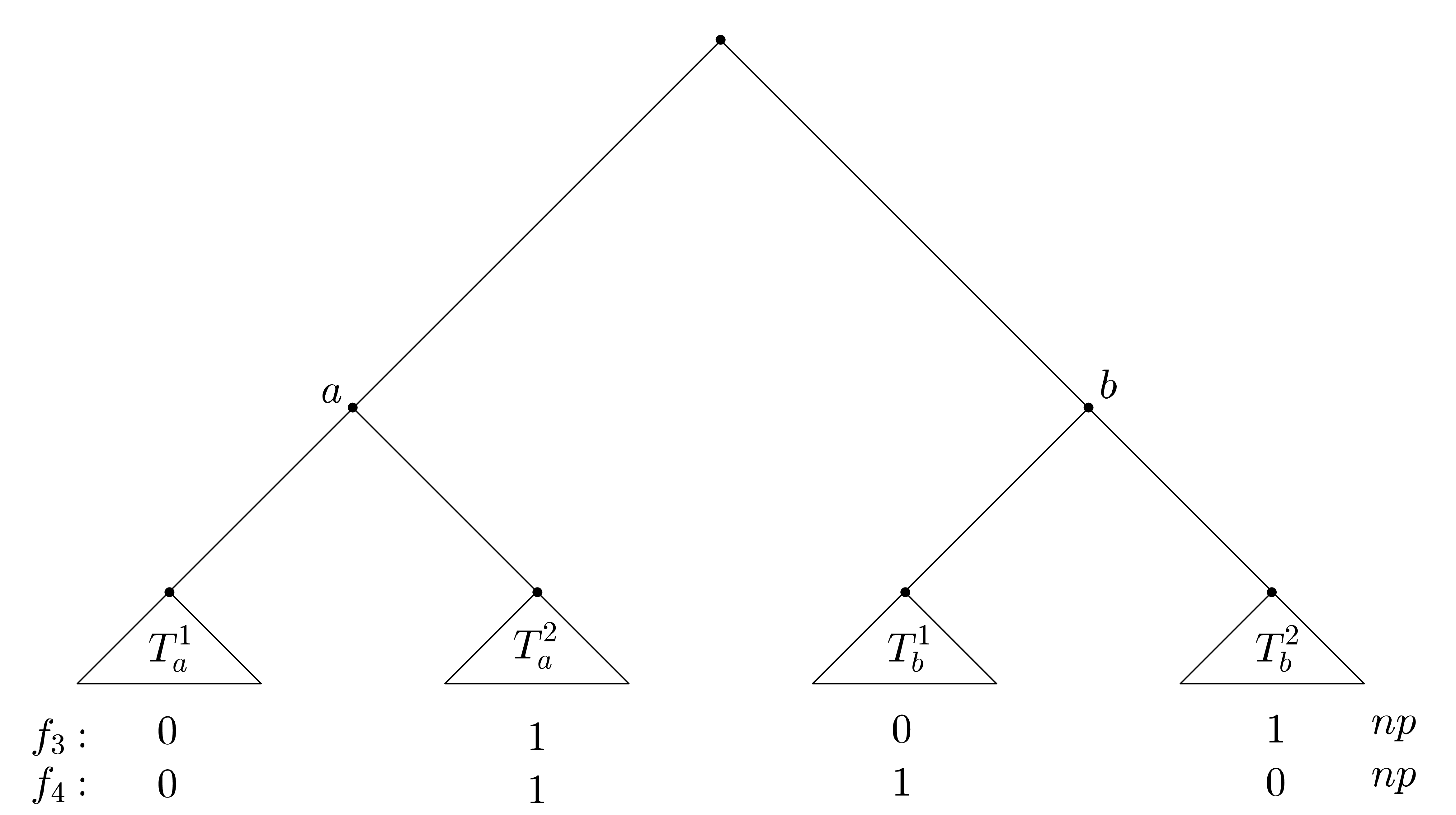}
	\caption{\blue{Phylogenetic $X$-}tree $T$ together with characters $f_3$ and $f_4$.}
	\label{Fig_4suffice_6}
\end{figure}

Now, notice that both $f_3$ and $f_4$ exclude the root position on the edge leading to $x_1$, because they would be both persistent on a tree rooted at this edge (Figure \ref{Fig_4suffice_7}). Thus, in the case that $T_b$ has at least two leaves and we have characters $f_3$ and $f_4$, we do not need character $f_2$ anymore to exclude this root position.

\begin{figure}[htbp]
	\centering
	\includegraphics[scale=0.1]{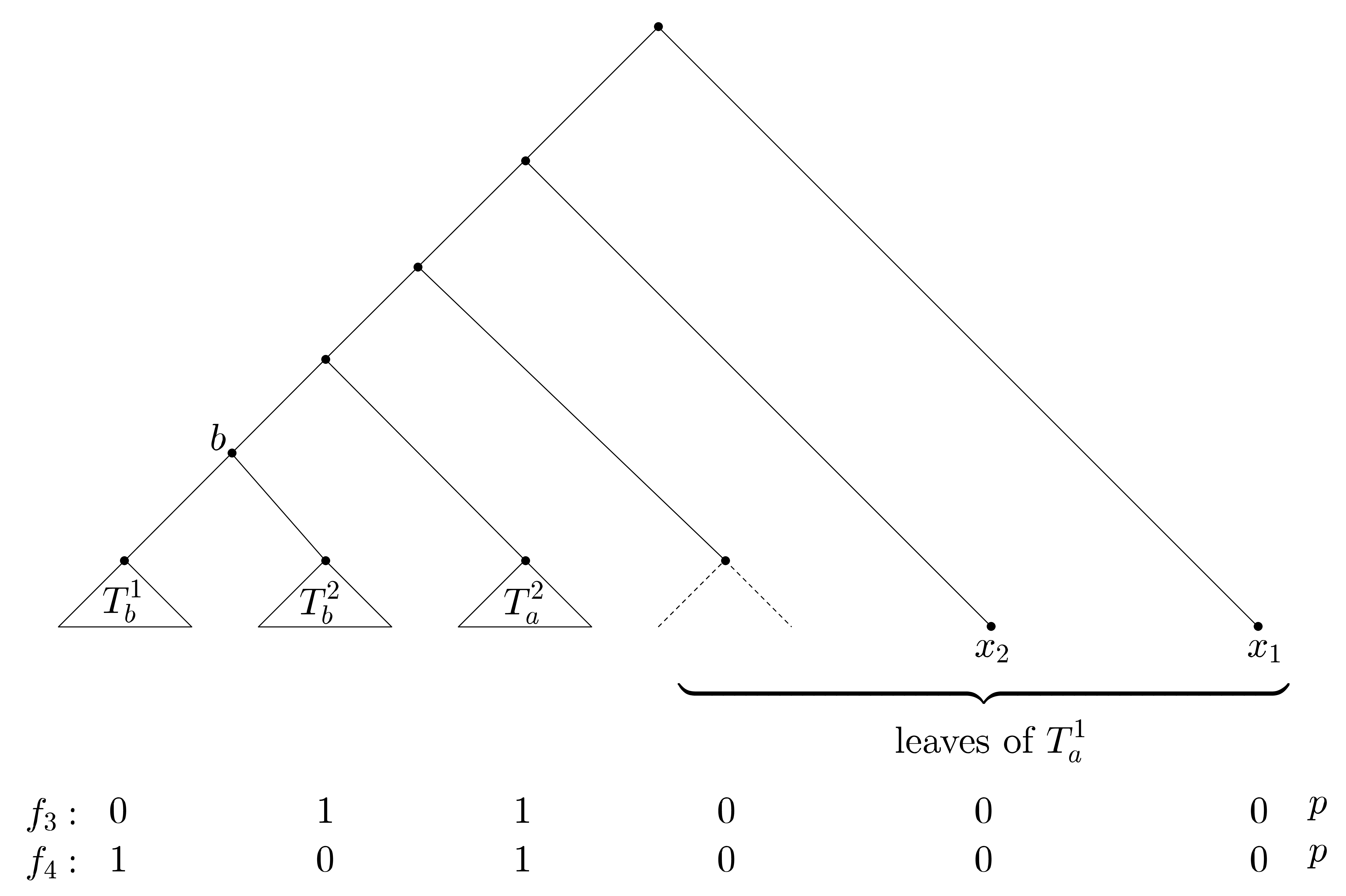}
	\caption{\blue{Phylogenetic $X$-}tree $T^u$ as depicted in Figure \ref{Fig_4suffice_2} rooted on the edge leading to leaf $x_1$ together with characters $f_3$ and $f_4$.}
	\label{Fig_4suffice_7}
\end{figure}

Moreover, $f_3$ and $f_4$ also exclude all edges in $T_b$ as root positions.
For illustration, $f_3$ on $T^u$ is depicted by Figure \ref{Fig_4suffice_4}. Now if the root position was placed somewhere in $T_b^1$ or on the edge leading from $b$ to $T_b^1$, $f_3$ would be persistent on $T$, as can be seen in Figure \ref{Fig_4suffice_5}. But we know that the persistence status of $f_3$ is $np$, so all these edges can be excluded. 

\begin{figure}[htbp]
	\centering
	\includegraphics[scale=0.1]{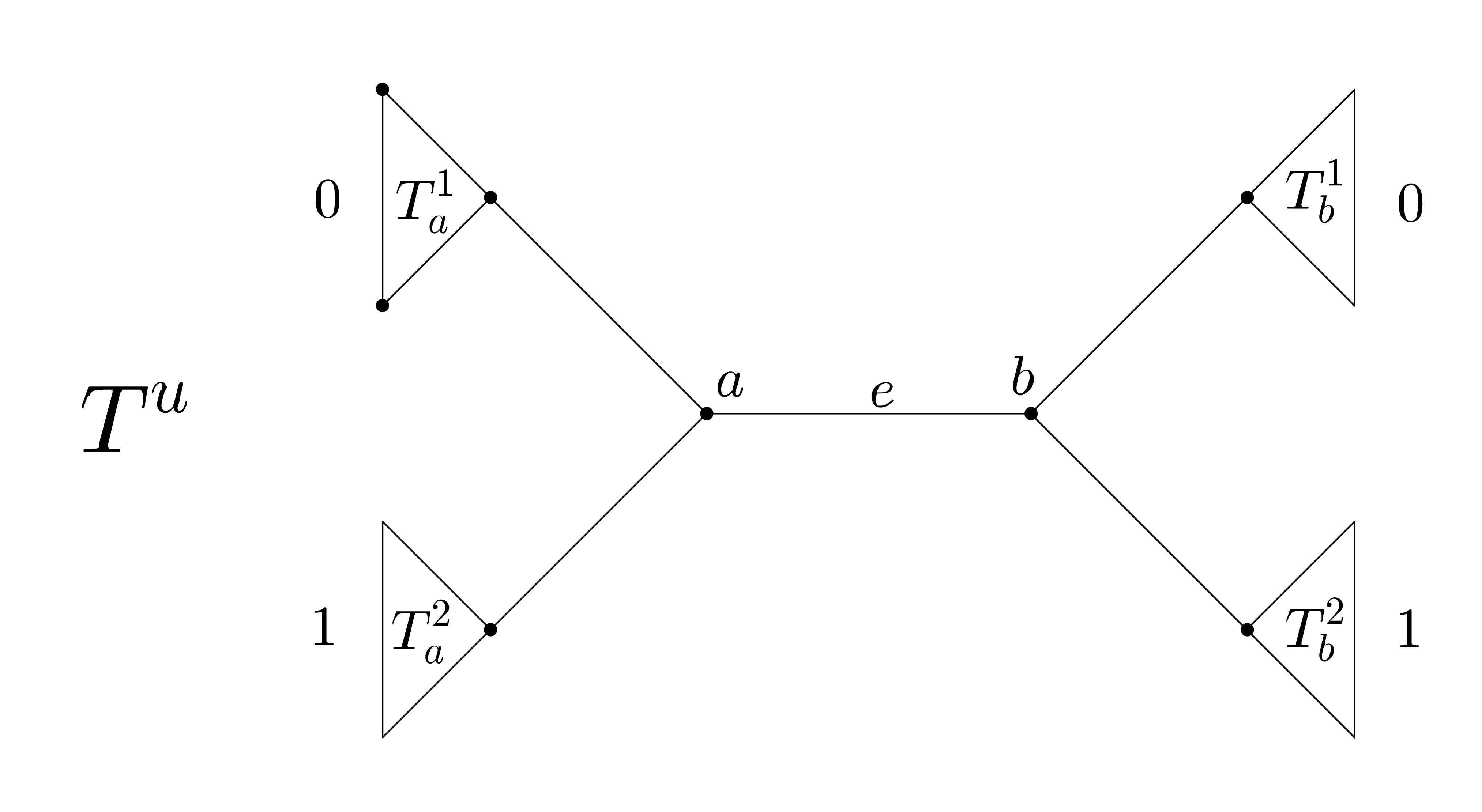}
	\caption{Character $f_3$ on $T^u$ when $T_b$ has more than one leaf.}
	\label{Fig_4suffice_4}
\end{figure}

\begin{figure}[htbp]
	\centering
	\includegraphics[scale=0.1]{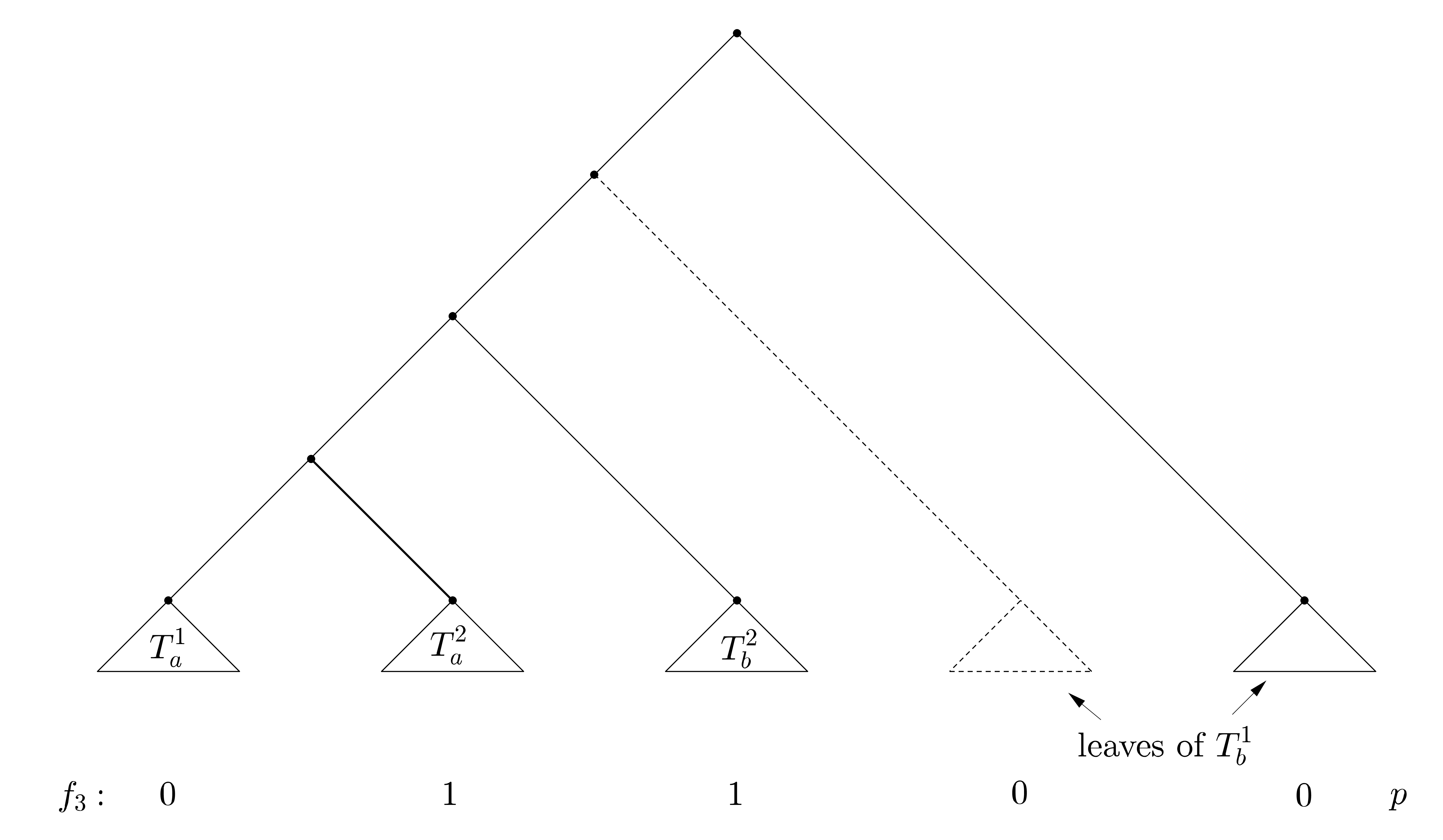}
	\caption{$T^u$ as depicted in Figure \ref{Fig_4suffice_4} rooted somewhere in $T_b^1$ or on the edge leading from $b$ to $T_b^1$ together with the character $f_3$. Note that $f_3$ is persistent on this rooting, while it is not persistent on $T$ as depicted in Figure \ref{Fig_4suffice_6}. }
	\label{Fig_4suffice_5}
\end{figure}

Using $f_4$, we can analogously exclude all edges in $T_b^2$ as well as the one leading from $b$ to $T_b^2$. So in total, all edges in $T_b$ can now be excluded. Thus, the only remaining option for the root position is edge $e=\{a,b\}$, which induces the correct rooted tree $T$. 

Summarizing the above, this completes the proof, as we have shown that three characters together with their persistence status suffice to determine the correct root position. If $T_b$ contains only one leaf, we can use $f_1$ and $f_2$ to determine the correct root position. If $T_b$ contains at least two leaves, we use $f_1, \, f_3$ and $f_4$.
\end{enumerate}
\end{proof}

We now illustrate Lemma \ref{4suffice} with two examples. 
\begin{example} First, consider \blue{the phylogenetic $X$-}tree $T_1^{\ast}$ on 5 leaves depicted in Figure \ref{Fig_Example1}. Decomposing $T_1^{\ast}$ into its two maximal pending subtrees $T_a$ and $T_b$ leads to the situation that $T_b$ consists of only one leaf. Thus, as indicated in the proof of Lemma \ref{4suffice}, two characters, namely $f_1 = 01110$ and $f_2=10110$, together with their persistence status \red{(in this case: $\mathcalligra{p}(f_1,T_1^\ast)=\mathcalligra{p}(f_2,T_1^\ast)=p$)} suffice to distinguish $T_1^{\ast}$ from any other \blue{phylogenetic $X$-}tree $\widetilde{T}$ with $\widetilde{T}^u = T_1^{{\ast}u}$. 

Now, consider \blue{the phylogenetic $X$-}tree $T_2^{\ast}$ on 5 leaves depicted in Figure \ref{Fig_Example2}. In this case, both maximal pending subtrees contain more than one leaf and three characters together with their persistence status suffice to distinguish $T_2^{\ast}$ from any other \blue{phylogenetic $X$-}tree $\widetilde{T}$ with $\widetilde{T}^u = T_2^{{\ast}u}$, namely \red{$f_1=01100$ (with $\mathcalligra{p}(f_1, T_2^\ast)=p$), $f_3=00101$ (with $\mathcalligra{p}(f_3, T_2^\ast)=np$ ), and $f_4=00110$ (with $\mathcalligra{p}(f_4, T_2^\ast)=np$)}.

\begin{figure}[htbp]
	\centering
	\includegraphics[scale=0.22]{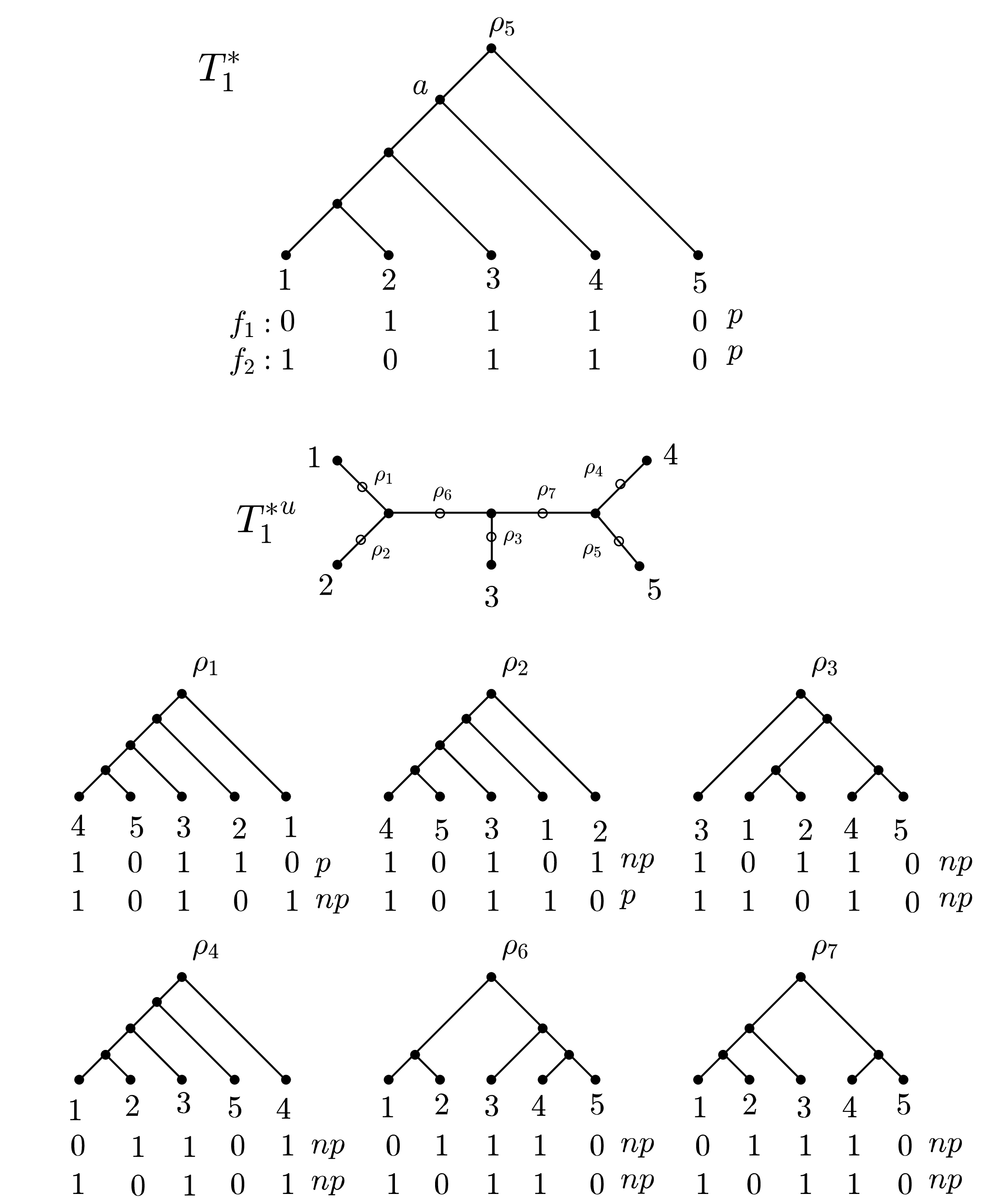}
	\caption{\blue{Phylogenetic $X$-}tree $T_1^{\ast}$ and its unrooted version $T_1^{{\ast}u}$. The two characters $f_1$ and $f_2$ together with their persistence status suffice to distinguish $T_1^{\ast}$ from any of the other rootings of $T_1^{{\ast}u}$, because on none of the six other rootings $f_1$ and $f_2$ are \emph{both} persistent.}
	\label{Fig_Example1}
\end{figure}

\begin{figure}[htbp]
	\centering
	\includegraphics[scale=0.2]{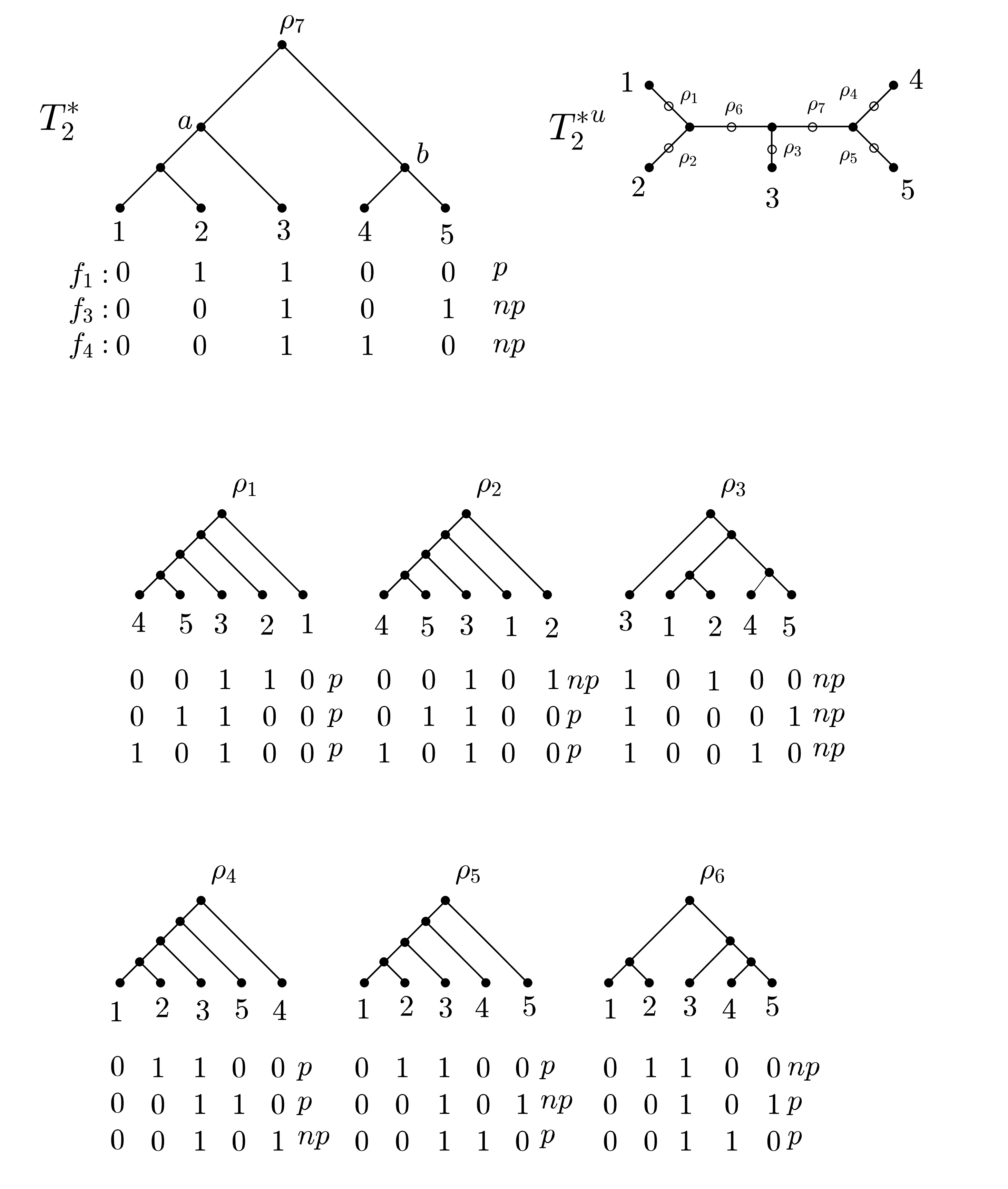}
	\caption{\blue{Phylogenetic $X$-}tree $T_2^{\ast}$ and its unrooted version $T_2^{{\ast}u}$. The three characters $f_1, f_3$ and $f_4$ together with their persistence status suffice to distinguish $T_2^{\ast}$ from any of the other rootings of $T_2^{{\ast}u}$, because on none of the other rootings we have that $f_1$ is persistent, while $f_3$ and $f_4$ are not persistent.}
	\label{Fig_Example2} 
\end{figure}

\end{example}

We are now in the position to prove that $2(n-3)+3 = 2n-3$ characters suffice to uniquely determine a tree $T$, which was the claim of Theorem \ref{bound}.

\begin{proof}[Proof of Theorem \ref{bound}] By Corollary \ref{2n-3suffice}, we can fix the unrooted version $T^u$ of $T$ using $2(n-3)$ characters and their persistence status. Then, by Lemma \ref{4suffice} we know that, given $T^u$, three characters with their persistence status suffice to fix $\rho$. This completes the proof.
\end{proof}

Note, however, that $2n-3$ is only an upper bound for the number of characters together with their persistence status needed to uniquely determine a tree $T$ and does not have to be tight. We will elaborate on this in the discussion section.

\section{Discussion}
The \emph{perfect phylogeny with persistent characters} has recently been thoroughly studied from an algorithmic and computational perspective (e.g. \citet{Bonizzoni2012, Bonizzoni2014, Bonizzoni2017b, Bonizzoni2017}). 
Here, we have considered persistent characters from a combinatorial point of view and have analyzed different aspects.  

First of all, we have illustrated the connection between persistent characters and the principle of Maximum Parsimony. In particular, we have established a connection between persistence and the first phase of the Fitch algorithm. Based on the Fitch algorithm it can easily be decided whether a binary character $f$ is persistent on a given rooted binary phylogenetic tree $T$.

We have then turned to the question of how many characters are persistent on a given tree $T$. We have seen that this quantity depends on the tree shape of $T$ and we could show that the number of persistent characters $\mathcal{P}(T)$ can be derived from the so-called Sackin index of $T$. 
In principle, the more balanced a tree is (in terms of the Sackin index), the fewer persistent characters it has. 

The last aim of our manuscript was then to determine the number of (carefully chosen) binary characters together with their persistence status that uniquely determine a phylogenetic tree $T$. We have shown that this number is bounded from above by $2n-3$, where $n$ is the number of leaves of $T$. 
Note, however, that this only provides a rough upper bound on the minimum number of characters needed together with their persistence status in order to uniquely determine $T$. We used Mathematica \citet{Mathematica} to perform an exhaustive search through tree space for $n=4$, $n=5$ and $n=6$, respectively. This search yielded that $3$, $4$ and $6$ characters, respectively, together with their persistence status are sufficient to determine all possible trees. Thus, we conjecture that the bound suggested by Corollary \ref{bound}, which would have been 5, 7 and 9, respectively, is not tight for any value of $n$. We are therefore planning to establish an improved bound in a subsequent study.

\section*{Acknowledgements}
The authors wish to thank Remco Bouckaert for helpful discussions and two anonymous reviewers for valuable suggestions on an earlier version of this manuscript. Moreover, Kristina Wicke thanks the German Academic Scholarship Foundation for a studentship, and Mareike Fischer thanks the joint research project \textit{\textbf{DIG-IT!}} supported by the European Social Fund (ESF), reference: ESF/14-BM-A55-0017/19, and the Ministry of Education, Science and Culture of Mecklenburg-Vorpommern, Germany.

\bibliographystyle{model1-num-names}\biboptions{authoryear}
\bibliography{References}

\section{Appendix}
\subsection*{Supplementary Results}

\setcounter{lemma}{1}
\begin{lemma}
Let $f$ be persistent on $T$ and let $g$ be a minimal persistent extension  of $f$ on $T$. Then, if $g$ contains a $0 \rightarrow 1$ change on an edge $(u,v)$ and a subsequent $1 \rightarrow 0$ change on an edge $(w,x)$, these two change edges are not adjacent, i.e. $v \neq w$. 
\end{lemma}

\begin{proof} We assume that the statement does not hold, i.e. we assume $g$ is a minimal persistent extension of $f$ on $T$ with two change edges $(u,v)$ ($0 \rightarrow 1$ change) and $(v,x)$ ($1 \rightarrow 0$ change), i.e. the two change edges share a common node $v$ (Figure \ref{Fig_notadjacent}). Then, because $f$ is persistent and $T$ is binary, $v$ has another direct descendant $x'$, which only has descending leaves of state 1. This is due to the fact that $x'$ is contained in the subtree rooted at $v$ (so the 1 has already been gained), but not in the subtree rooted at $x$ (so the 1 has not yet been lost). Note that by construction, these leaves of the subtree rooted at $x'$ are the only leaves in state 1 in $T$, i.e. all other leaves are in state 0. So if we construct an extension $\tilde{g}$ (Figure \ref{Fig_notadjacent}) with all nodes in the subtree rooted at $x'$ in state 1 and all other nodes in state 0, this extension $\tilde{g}$ has a $0 \rightarrow 1$ change on the edge $(v,x')$ but no other changes. So $\tilde{g}$ is a persistent extension of $f$ on $T$, but requires only one change. This contradicts the minimality of $g$, which completes the proof.
\end{proof}

\begin{figure} [htbp]
	\centering
	\includegraphics[scale=0.1]{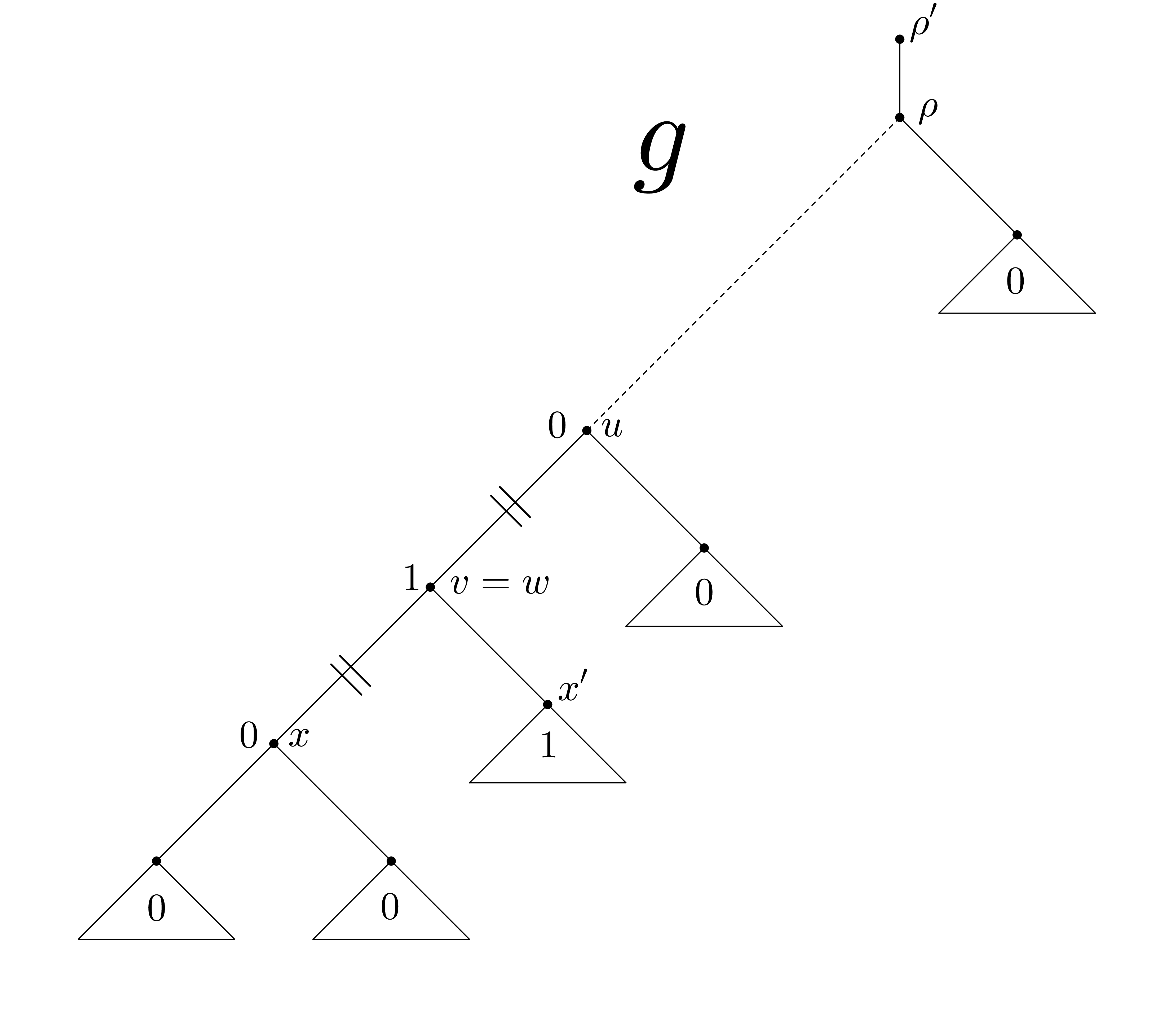}
	\includegraphics[scale=0.1]{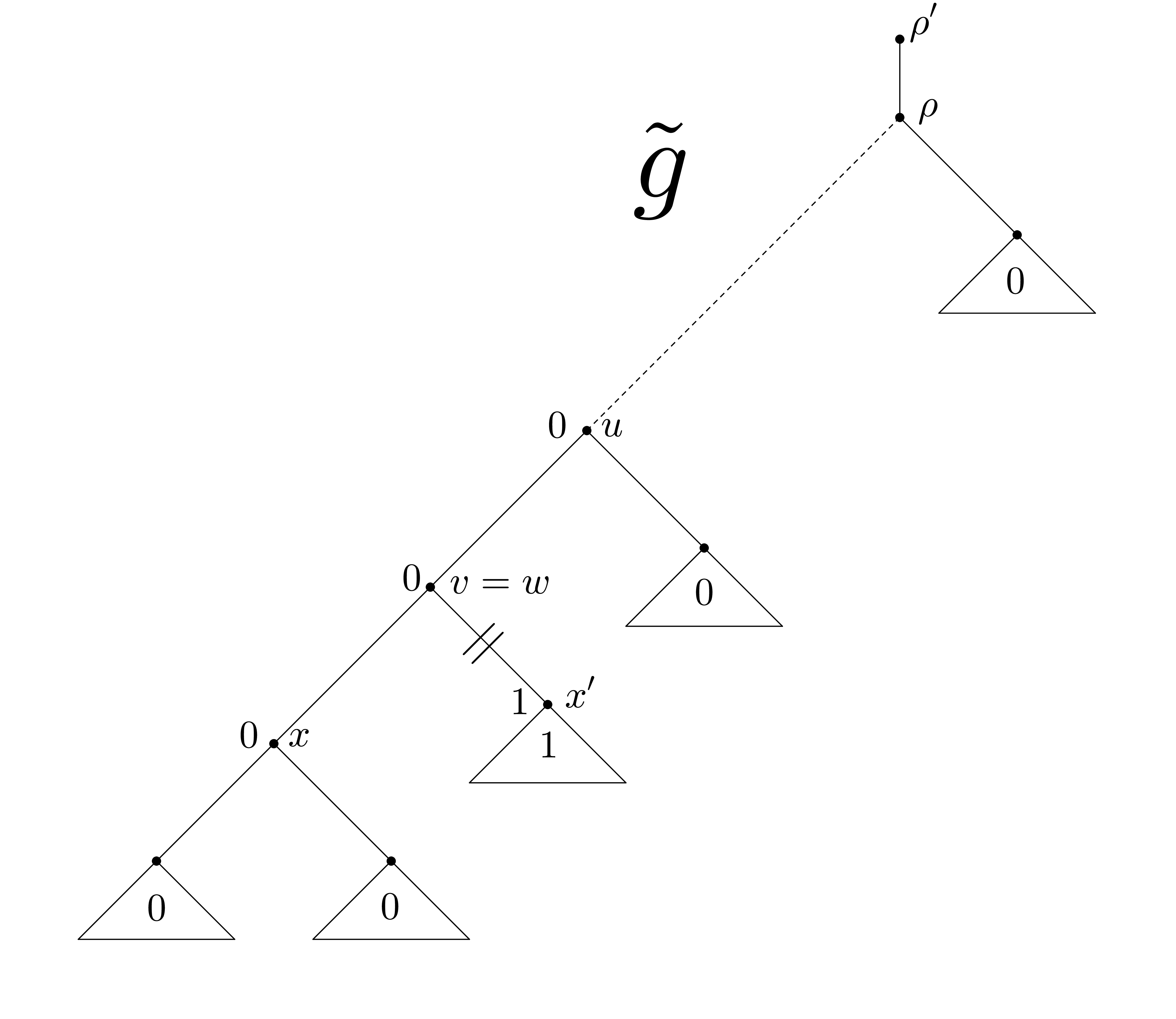}
	\caption{Extensions $g$ and $\tilde{g}$ considered in the proof of Lemma \ref{notadjacent}. In $g$ the $0 \rightarrow 1$ change and the $1 \rightarrow 0$ change are adjacent and lead to two changes. As $\tilde{g}$ requires only one $0 \rightarrow 1$ change, $\tilde{g}$ leads to a persistence score of 1. Thus, $g$ is not a minimal persistent extension.}
	\label{Fig_notadjacent}
\end{figure}

\setcounter{lemma}{2}
\begin{lemma}Let $f$ be persistent on $T$ and $l(f,T)=2$ and let $g$ be a minimal persistent extension of $f$ on $T$. Then, $g$ has the property that all of its change edges have a source node that is assigned $\{0,1\}$ by the first phase of the Fitch algorithm. 
\end{lemma}

\begin{proof}
Let $g$ be a minimal persistent extension of $f$ on $T$. As $l(f,T)=2$ and $f$ is persistent, we conclude that $g$ contains precisely two change edges, namely one $0 \rightarrow 1$ change edge, say $(u,v)$, and one $1 \rightarrow 0$ edge, say $(w,x)$, which by Lemma \ref{notadjacent} do not share a common node. We need to show that the first phase of the Fitch algorithm will assign state set $\{0,1\}$ to both $u$ and $w$ (Figure \ref{fig_01nodes}).

We first consider $w$. By definition of persistence, all leaves descending from $x$ are in state 0 (as there cannot be any more changes after the $1 \rightarrow 0$ change). We refer to the other child node of $w$ as $x'$. Now, all leaves descending from $x'$ are in state 1, because they are part of the subtree rooted at $v$, and are thus subject to the $0 \rightarrow 1$ change, but they are not part of the subtree rooted at $x$ and are thus not subject to the $1 \rightarrow 0$ change. But as all leaves descending from $x$ are in state 0, the first phase of Fitch will assign state set $\{0\}$ to $x$. Likewise, as as all leaves descending from $x'$ are in state 1, the first phase of Fitch will assign state set $\{1\}$ to $x'$. Thus, $w$ will be assigned $\{0,1\}$ by the first phase of the Fitch algorithm.

Now we consider $u$. We know that $g(u)=0$ and $g(v)=1$, and we know that $u$ has another descendant $v'$ with all leaves descending from $v'$ being in state 0. This is due to the fact that $g$ is a persistent extension of $f$, and thus all leaves of $T$ that are not descending from the $0 \rightarrow 1$ edge $(u,v)$ must be in state 0. So as all leaves descending from $v'$ are in state 0, the first phase of the Fitch algorithm will assign state set $\{0\}$ to $v'$. 

Next, note that $v$ has a direct descendant $z$ on the path from $v$ to $w$ (note that $z$ might be equal to $w$), and another direct descendant $z'$. Note that all leaves descending from $z'$ must be in state $1$, as they belong to the subtree rooted at $v$, i.e. they come `after' the $0 \rightarrow 1$ change, but they do not belong to the subtree rooted at $x$, i.e. they come `before' the $1 \rightarrow 0$ change. So $z'$ will be assigned state set $\{1\}$ by the first phase of the Fitch algorithm. 

Note that $z$ has a child node $y$ which is not on the path from $z$ to $w$, and all leaves descending from $y$ are in state 1 (if $z=w$, then $y=x'$). This is again due to the fact that $y$ belongs to the subtree rooted at $v$, i.e. $y$ comes `after' the $0 \rightarrow 1$ change, but it does not belong to the subtree rooted at $x$, i.e. $y$ comes `before' the $1 \rightarrow 0$ change. So $y$ will be assigned state set $\{1\} $ by the first phase of the Fitch algorithm, which implies that $z$ can only be assigned $\{1\} $ or $\{0,1\} $ (if $z=w$, we already know that $z$ is assigned $\{0,1\}$). In both cases, as we have already shown that $z'$ will be assigned state set $\{1\}$ and as we have $\{1\} \cap \{0,1\} =\{1\} \cap \{1\} =\{1\} $, $v$ will be assigned state set $\{1\}$. 

So altogether we know that the children of $u$, namely $v$ and $v'$, will be assigned $\{1\}$ and $\{0\}$, respectively, which implies that $u$ must be assigned $\{0,1\}$ by the first phase of the Fitch algorithm. This completes the proof.

\end{proof}

\begin{figure} [htbp]
	\centering
	\includegraphics[scale=0.1]{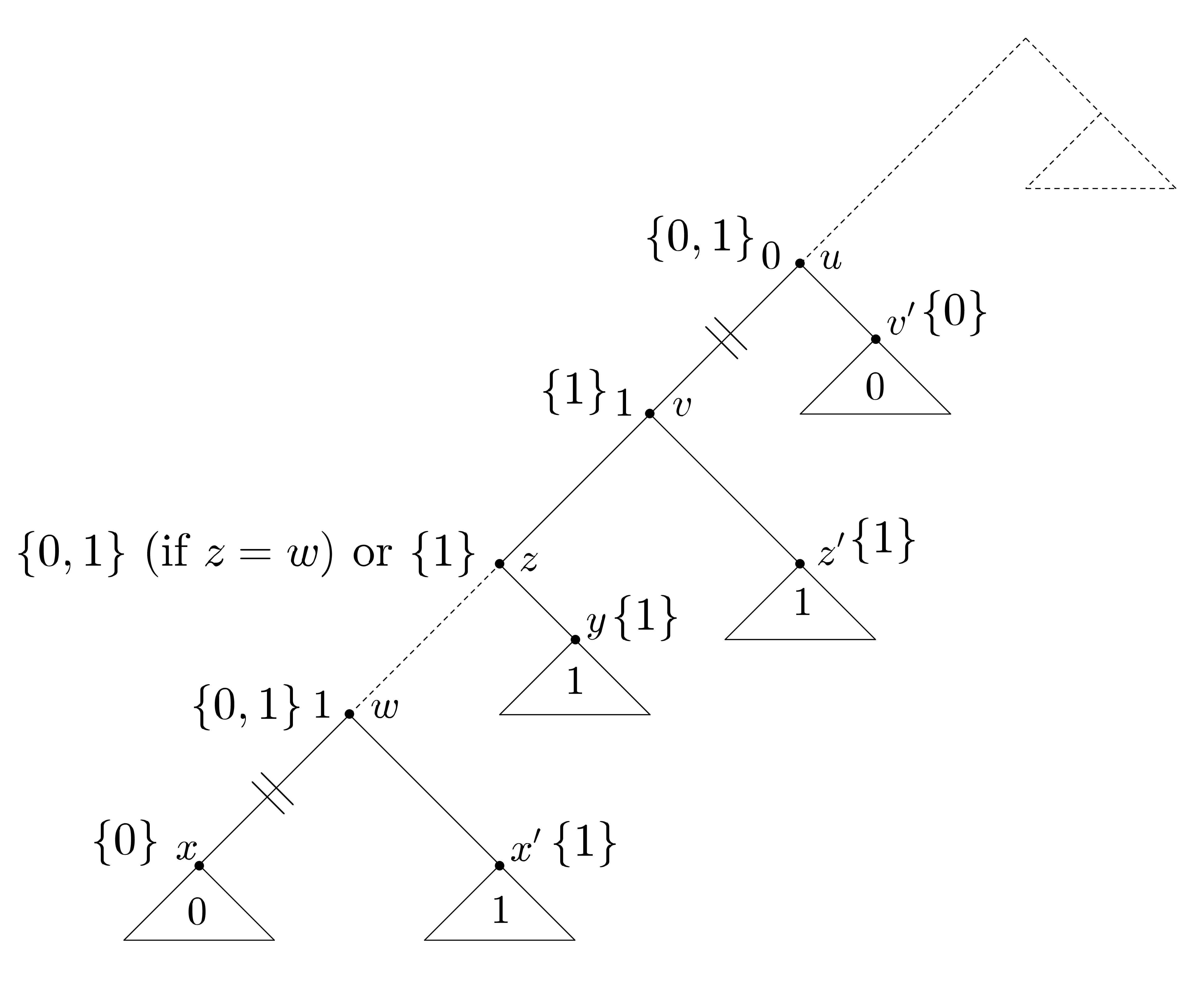}
	\caption{Situation in Lemma \ref{01nodes}. There are two change edges, namely a $0 \rightarrow 1$ change edge $(u,v)$ and a $1 \rightarrow 0$ change edge $(w,x)$, which do not share a common node. All leaves descending from $x$ are in state 0, while all leaves descending from $x'$ are in state 1. Moreover, all leaves that are not descending from $v$ (and are thus not affected by the $0 \rightarrow 1$ change) must be in state 0, while all nodes that are descending from $v$ but not from $w$ must be in state 1. The first phase of the Fitch algorithm, thus necessarily assigns state set $\{0,1\}$ to nodes $w$ and $u$.}
	\label{fig_01nodes}
\end{figure}

\setcounter{theorem}{0}
\begin{theorem}[Characterization of persistent characters] 
Let $f$ be a binary character on $\mathcal{R}=\{0,1\}$ and let $T$ be a phylogenetic tree. Then, we have:
\begin{enumerate}
\item If $l(f,T) > 2$, then $f$ is not persistent on $T$.
\item If $l(f,T) \leq 1$, then $f$ is persistent on $T$. 
\item If $l(f,T) = 2$, let the two $\{0,1\}$ union nodes found during the 1st phase of the Fitch algorithm be denoted by $u$ and $w$, respectively. Then, we have:
 $$ f \mbox{ is persistent on } T $$ 
 $$\Leftrightarrow $$
\hspace*{3cm} all of the following conditions hold: \\

\begin{enumerate} 
\item $u$ is an ancestor of $w$ or vice versa; wlog $u$ is the ancestor of $w$. 
\item The ancestral state sets found by the first phase of the Fitch algorithm fulfill the following conditions (Figure \ref{Fig_characterize}):
	\begin{itemize}
	\item all nodes that are descendants of the direct descendant $v$ of $u$ on the path to $w$, but not of $w$ are assigned state set $\{1\}$ (in particular, all nodes on the path from $v$ to $w$ (including $v$) are assigned state set $\{1\}$),
	\item all nodes that are not descendants of $v$ are assigned state set $\{0\}$.
	\end{itemize}
\end{enumerate}
\end{enumerate}
\end{theorem}

\begin{proof}[Proof of Theorem \ref{characterize}, Parts 1 and 2]  \leavevmode
\begin{enumerate}
\item Let $l(f,T) > 2$. Then we use the second part of Lemma \ref{MPboundsPersistence} to conclude that $f$ is not persistent on $T$.

\item Now, let $l(f,T) \leq 1$.
If $l(f,T)=0$, then $f$ equals either $0,\ldots,0$, in which case $l_p(f,T)=0$ (we can simply assign all nodes of $T$ state 0), or $f$ equals $1,\ldots,1$, in which case $l_p(f,T)=1$, as one $0\rightarrow 1$ change is needed on the root edge (i.e. all nodes of $T$ can be assigned state 1). In both cases, as there exists a persistent extension of $f$, $f$ is persistent.

Now, if $l(f,T)=1$, this implies that there is only one internal node of $T$ (excluding $\rho'$) that is assigned state set $\{0,1\}$, while all other internal nodes are assigned single states $\{0\}$ or $\{1\}$.  This necessarily implies that the two maximal pending subtrees of the node assigned $\{0,1\}$ are such that one of them only contains leaves in state 0 and the other one only contains leaves in state 1. We now construct a persistent extension of $f$ on $T$. For all nodes for which the 1st phase of the Fitch algorithm makes an unambiguous choice, i.e. $\{0\}$ or $\{1\}$, we assign the corresponding state to the respective node. So now we only have to decide what to do with the $\{0,1\}$ node, and we have to show that the resulting extension is persistent.

So if the $\{0,1\}$ node is $\rho$, we choose the state of $\rho$ to be 0 and require one $0 \rightarrow 1$ change on the edge leading to the maximal pending subtree with all nodes in state 1. Thus, this extension $g$ of $f$ is persistent, and therefore $f$ is persistent on $T$.

On the other hand, if the $\{0,1\}$ is assigned to some inner node other than $\rho$, the root is already either uniquely assigned state 0 or 1. If the root is in state 0, we can choose the $\{0,1\}$ node to be in state 0. Again, we require exactly one $0 \rightarrow 1$ change leading to the maximal pending subtree of the $\{0,1\}$ node that has only leaves in state 1. Thus, this extension $g$ and therefore also $f$ is persistent on $T$.

Now the only case that remains is the case where $\rho$ is assigned state 1 and $\{0,1\}$ is assigned to some other internal node. In this case, we can choose the $\{0,1\}$ node to be in state 1, which implies that we have one $1 \rightarrow 0$ change on the edge leading to its maximal pending subtree with all nodes in state 0. Additionally, we require one $0 \rightarrow 1$ change on the root edge. Again, the resulting extension $g$ and therefore also $f$ is persistent. This completes the proof.

\end{enumerate}
\end{proof}

\setcounter{lemma}{5}

\begin{lemma} Let $T$ be a rooted binary phylogenetic tree on $n \geq 2$ leaves. For each inner node $u$ of $T$, i.e. $u \in \mathring{V}(T)$, let $d_u$ denote the number of children of $u$ that are also inner nodes of $T$, i.e. $d_u$ can assume values 0, 1 or 2. Then, we have: $$\sum\limits_{u \in \mathring{V}(T)} d_u = n-2.$$
\end{lemma}

\begin{proof}[Proof of Lemma \ref{d_u}] We prove this by induction on $n$. For $n=2$, there is only one possible tree, which consists only of one inner node, namely the root $\rho$, that is connected to both leaves and thus has $d_\rho=0$. So in this case, we have  $\sum\limits_{u \in \mathring{V}(T)} d_u = d_\rho=0=2-2=n-2. $ This completes the base case of the induction.

Now let us assume the statement holds for trees with up to $n$ leaves. Consider a tree $T$ with $n+1$ leaves, and assume without loss of generality that $n+1\geq 3$ (otherwise, if $T$ has only two leaves, consider the base case again). We now need to show that for $T$, we have $\sum\limits_{u \in \mathring{V}(T)} d_u =(n+1)-2$. 

Recall that every rooted binary phylogenetic tree on at least two leaves has at least one cherry (see for example \citet[p. 9]{Book_Steel}). So $T$ has at least one cherry, whose leaves we call $x$ and $y$, respectively. Now $x$ and $y$ have a direct ancestor, which we call $a$. Note that $d_a=0$ as $a$ is connected to two leaves.

However, now we create a tree $T'$ by deleting leaves $x$ and $y$ together with edges $(a,x)$ and $(a,y)$ (Figure \ref{Fig_du}). This implies that $a$ is now a leaf, so $T'$ has now $(n+1)-2+1=n$ leaves. Thus, by the inductive hypothesis,  $\sum\limits_{u \in \mathring{V}(T')} d_u = n-2$. 

Now consider $\sum\limits_{u \in \mathring{V}(T)} d_u$: In fact, this sum only differs from $\sum\limits_{u \in \mathring{V}(T')} d_u$ in two ways: first, it has one more summand, namely $a$, but this does not contribute to the sum because $d_a=0$. Second, as we know that $T$ has at least three leaves, we know that $a$ has an ancestor, say $b$, which is in $T$ connected to an inner node (namely $a$), but which in $T'$ is instead connected to a leaf (as $a$ is now a leaf). So $d_b(T)=d_b(T')+1$. All other nodes of $T'$ remain by construction unchanged compared to $T$. Thus, in total we conclude:

$$\sum\limits_{u \in \mathring{V}(T)} d_u = \sum\limits_{u \in \mathring{V}(T')} d_u +1 \stackrel{ind.}{=}(n-2)+1=(n+1)-2.$$

This completes the proof.
\end{proof}

\begin{figure}[htbp]
	\centering
	\includegraphics[scale=0.1]{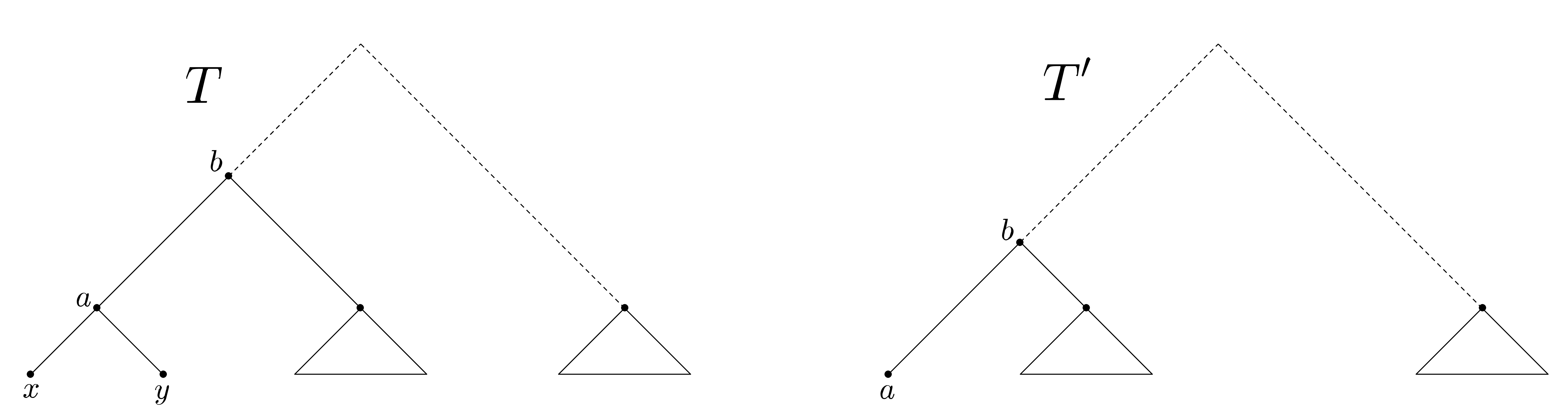}
	\caption{Trees $T$ and $T'$ used in the proof of Lemma \ref{d_u}.}
	\label{Fig_du}
\end{figure}

\subsection*{Algorithm}
\begin{algorithm}[H]
\KwIn{Rooted binary phylogenetic tree $T$, binary character $f$}
\KwOut{Persistence status $\mathcalligra{p}(f,T)$ of $f$ ($p$ if $f$ is persistent on $T$, $np$ if it is not)}
Use the first phase of the Fitch algorithm in order to compute ancestral state sets for all internal nodes of $T$ and the parsimony score $l(f,T)$ \;
\uIf{$l(f,T) \leq 1$}{
	return $f$ is persistent, i.e. $\mathcalligra{p}(f,T) = p$ \;
	}
\uElseIf{$l(f,T) \geq 3$}{
	return $f$ is not persistent, i.e. $\mathcalligra{p}(f,T) = np$\;
	}
\uElseIf{$f(l,T)=2$}{
	Let $u$ and $w$ be the two union nodes with ancestral state set $\{0,1\}$\;
	\uIf{$u$ is an ancestor of $w$ or $w$ is an ancestor of $u$}{
		Assume $u$ is the ancestor of $w$ (otherwise exchange the labels of $u$ and $w$)\;
		\uIf{$T$ does not contain the edge $(u,w)$}{
			Consider the child $v\neq w$ of $u$ on the path from $u$ to $w$\;
			\uIf{\begin{tabular}{@{\hspace*{2mm}}l@{}}
			\hspace*{2mm} $\bullet$ all nodes on the path from $v$ to $w$ (including $v$) are assigned \\ state set $\{1\}$ by the first Fitch phase and \\
				\hspace*{2mm}  $\bullet$ all nodes that are not descendants of $v$ are assigned state set $\{0\}$ \\
				 \end{tabular}}{
				return $f$ is persistent, i.e. $\mathcalligra{p}(f,T) = p$\;
				}
			\Else{
				return $f$ is not persistent, i.e. $\mathcalligra{p}(f,T) = np$\;
			}
		}
		\Else{
			return $f$ is not persistent, i.e. $\mathcalligra{p}(f,T) = np$\;
		}
	}
	\Else{
		return $f$ is not persistent, i.e. $\mathcalligra{p}(f,T) = np$\;
	}
}
\caption{Calculate persistence status} 
\label{CheckPersistence}
\end{algorithm}

\end{document}